\documentclass{amsart}
\usepackage{amssymb}
\usepackage{verbatim}
\usepackage{mathrsfs}
\usepackage{dsfont}

\newtheorem{theorem}{Theorem}[section]
\newtheorem{lemma}[theorem]{Lemma}
\newtheorem{corollary}[theorem]{Corollary}
\newtheorem{proposition}[theorem]{Proposition}
\newtheorem{conjecture}[theorem]{Conjecture}

\theoremstyle{definition}
\newtheorem{definition}[theorem]{Definition}
\newtheorem{example}[theorem]{Example}

\theoremstyle{remark}
\newtheorem{remark}[theorem]{Remark}

\newcommand{\I}{{\mathds {1}}}
\newcommand{\cA}{{\mathcal A}}

\newcommand{\cH}{{\mathcal H}}

\newcommand{\cK}{{\mathcal K}}
\newcommand{\cM}{{\mathcal M}}
\newcommand{\qM}{{\mathfrak M}}

\newcommand{\cN}{{\mathcal N}}
\newcommand{\cP}{{\mathcal P}}

\newcommand{\cO}{{\mathcal O}}

\newcommand{\SOL}{{\mathrm{SOL}}}

\newcommand{\Rn}{{\rm I\!R}} 
\newcommand{\Cn}{{\setbox0=\hbox{
$\displaystyle\rm C$}\hbox{\hbox
to0pt{\kern0.6\wd0\vrule height0.9\ht0\hss}\box0}}} 

\numberwithin{equation}{section}

\begin{document}

\title[Von Neumann algebraic Quantum Theory on curved spacetime]{A von Neumann algebraic approach to Quantum Theory on curved spacetime}

\author{L. E. Labuschagne}
\address{DSI-NRF CoE in Math. and Stat. Sci,\\ Focus Area for Pure and Applied Analytics,\\ Internal Box 209, School of Math \& Stat. Sci.\\
NWU, PVT. BAG X6001, 2520 Potchefstroom\\ South Africa}
\email{Louis.Labuschagne@nwu.ac.za}

\author{W. A. Majewski}
\address{DSI-NRF CoE in Math. and Stat. Sci,\\ Focus Area for Pure and Applied Analytics,\\ Internal Box 209, School of Math. \& Stat. Sci.\\
NWU, PVT. BAG X6001, 2520 Potchefstroom\\ South Africa} 
\email{fizwam@gmail.com}

\date{\today}


\begin{abstract} By extending the method developed in our recent paper \cite{LM} we present the AQFT framework in terms of von Neumann algebras. 
In particular, this approach allows for a locally covariant categorical description of AQFT which moreover satisfies the additivity property and provides a 
natural and intrinsic framework for a description of entanglement. Turning to dynamical aspects of QFT we show that Killing local flows may be lifted to the 
algebraic setting in curved spacetime. Furthermore, conditions under which quantum Lie derivatives of such local flows exist are provided. The central question 
that then emerges is how such quantum local flows might be described in interesting representations. We show that quasi-free representations of Weyl algebras fit 
the presented framework perfectly. Finally, the problem of enlarging the set of observables is discussed. We point out the usefulness of Orlicz space 
techniques to encompass unbounded field operators. In particular, a well-defined framework within which one can manipulate such operators is necessary 
for the correct presentation of (semiclassical) Einstein's equation.

\end{abstract}

\maketitle

\noindent

\section{Introduction}

It is generally accepted that observables and states are basic concepts in the description of a physical system. The description of a quantum system which 
emphasises operators representing observables leads to the algebraic approach to Quantum Theory. For a motivation of such approach we refer the reader to the book 
of Emch \cite{Emch}. Then let us note that in Quantum Field Theory (QFT for short) there are two axiomatic programmes for QFT. The one which focuses on local 
observables and puts emphasis on operators is called Algebraic Quantum Field Theory (AQFT for short). The second one, more traditional, focuses on fields and 
Wightman functions and leads to the so called G{\aa}rding-Wightman programme, for details see the Haag book  \cite{Haag}.

The basic feature of AQFT is its emphasis on the description of algebraic relations between the primary objects of the theory. In other words, fields and 
observables are treated as purely algebraic (and abstract) objects. On the other hand, when examining a specific system in a specific situation, in addition to 
algebraic relations between objects, we must use the available knowledge about the system. This knowledge is described by the state of the system. Therefore, in 
the analysis of field systems, one should move to a representation that takes into account algebraic relations between the examined objects as well as the current 
state of the object (field). The implementation of such a program involves the use of the GNS (Gelfand-Naimark-Siegel) structure. As a result, we obtain a representation of the initial 
algebraic structure on a specific Hilbert space, with the considered representation reflecting the specific state of the described system.

Recalling that the result of a physical measurement of a specific observable is described by the value of the state of the system on that observable, 
reasonable topologies from the point of view of a physicist are weak topologies. Therefore, it is natural for an algebraic structure to be additionally equipped 
with a weak topology. Since von Neumann algebras satisfy the above conditions, they will be the focus of the following analysis. In concluding our general remarks 
about representations, let us add that, on moving to (faithful) representations, the set of states in the representation (folium) is a dense set in the set of 
states of the initial algebra, see Section III.2 in \cite{Haag}. Therefore, the transition to a representation is not in essence a limitation of the description.

In the introduction to \cite{BFR} the authors rightly note that since in AQFT all structures are formulated in terms of local quantities, it is inappropriate to 
demand the a priori existence of a distinguished global Hilbert space of states. The authors also note that whichever category of operator algebras is used to 
model the theory needs to be a tensorial category. By basing the theory on abstract $C^*$-algebras rather than von Neumann algebras, the authors at the outset 
exclude any difficulties regarding Hilbert spaces since none are required. As pointed out on pages 8 and 9 of the related paper \cite{FR}, this choice also 
circumvents problems related to generalizing the topological field theory construction to the case of infinite dimensional Hilbert spaces. Selecting 
$C^*$-algebras as the basic object may help to avoid some of the complexities involved incorporating Hilbert spaces into the picture but on the downside 
$C^*$-algebras allow for a menagerie of ``canonical'' tensor products. More seriously the exclusion of Hilbert spaces also excludes any chance of a refined theory 
of entanglement. 

The important point to note here is the fact that $C^*$-algebras can be considered as a noncommutative analogue of the family of continuous functions while von Neumann 
algebras can be considered as a noncommutative counterpart of measurable functions. This fact places von Neumann algebras at the heart of noncommutative integration theory. 
In our works, \cite{ML, Maj2, LM}, we proposed a description of large quantum systems (statistical physics, QFT) based on noncommutative analysis, i.e. on noncommutative 
integration theory and noncommutative differentiation. Since as noted above (abstract) von Neumann algebras are at the heart of noncommutative integration theory, we shall 
therefore use abstract von Neumann algebras when engaged in algebraic quantization. We emphasize that the use of abstract von Neumann algebras is not bound in any way to 
the choice of any particular state. See section \ref{S3} for details. 

One very obvious advantage of the von Neumann algebraic framework is the uniqueness of the von Neumann algebra tensor product. So the problem of 
having to decide which tensor product to use never arises. But this raises the question of how the requirement of not demanding a distinguished a priori Hilbert 
space is to be satisfied. We will demonstrate that the theory of the standard form of von Neumann algebras defines a tensorial category which enables one to 
incorporate Hilbert spaces into the picture at a local level. More specifically this approach will fulfil the criterion of not demanding the a priori existence of a 
distinguished global Hilbert space of states by relegating Hilbert spaces to a local phenomenon, thereby still leaving room for a theory of entanglement at a local level. It 
should be noted that the application of the standard form of von Neumann algebras to local algebras is nothing new, although to the best of our knowledge it doesn't seem to 
have done in the context of curved spacetime yet. As early as 1984 Doplicher and Longo already demonstrated the utility of these ideas in the context of Minkowski space - see 
\cite{DL}. 

Our approach to addressing the above challenges will be inductive in the sense that we will analyse the behaviour of local von Neumann algebras constructed using the free Klein-Gordon field with respect to these paradigms after which we will indicate how the emergent structure may be formalised into an axiomatic framework. Before more clearly outlining our objectives, we pause to demonstrate that the von Neumann algebra approach is still very much alive by drawing attention to some recent papers on this topic written from a von Neumann algebra perspective. Specifically the following: 

\begin{enumerate}
\item An algebra of observables for de Sitter space by Chandrasekaran, Longo, Penington, Witten \cite{CLPW}. Here the authors describe an algebra of observables for a static patch in de Sitter 
space, with operators gravitationally dressed to the worldline of an observer. The algebra is a von Neumann algebra of Type $mathrm{II}_1$. There is a natural notion of entropy for a state of such an algebra. 
In their paper the authors also point out the difficulty of isolating a concept of entropy for type III factors. As a contribution to this debate the authors of the present paper wish to point out 
that in one of their earlier papers \cite[Definition 5.3]{ML}, they identified a class of states who do admit a concept of entropy which proves to be a natural generalisation of von Neumann entropy.
\item Gravity and the crossed product by Witten\cite{Witt}: In the context of quantum gravity Witten here uses crossed products to give meaning to the notion of the $1/N$ expansion of the ambient 
algebra. Passing to the crossed product then enables him to introduce a notion of entropy for this system. One of the main benefits of passing to the crossed product (as also explained by Witten) is 
that it will transform the modular automorphism group - which is not inner for type III algebras - to an inner automorphism. The generator of this automorphism, which is some sort of Hamiltonian, is 
then affiliated to this enlarged algebra. The value of Witten's paper is that it shows that in many contexts crossed products find application in a natural way.
\begin{itemize}
\item We pause to remind the reader that crossed products werre similarly used in \cite{LM} to introduce the notion of $L^{\cosh-1}$-regularity which is a regularity restriction ensuring that for a 
large number of examples the field operators of Minkowski space may in a natural way be embedded into the Orlicz space $L^{\cosh-1}(\cM)$, which in turn may be viewed as a space which is home to 
``moment generating'' observables.
\end{itemize}
\item Algebras and states in JT gravity by Penington, Witten \cite{PW}. The authors analyze the algebra of boundary observables in canonically quantised JT gravity with or without matter. In the 
absence of matter, this algebra is commutative, generated by the ADM Hamiltonian. After coupling to a bulk quantum field theory, it becomes a highly noncommutative algebra of type $II_\infty$ with a 
trivial center.
\end{enumerate}

The above papers indicate that the von Neumann algebra approach still has strong support and that for at least for some models, it is to be preferred to the 
$C^*$-algebra approach. The above remarks should also indicate why in this article we choose to deal with AQFT based on von Neumann algebras.

The theme running throughout this paper is the suitability of von Neumann algebras for AQFT, in particular as regards the corpus of noncommutative analysis (specifically 
noncommutative integration theory and quantum Orlicz spaces) that von Neumann algebras in standard form bring to the table. This line of enquiry was initiated in our first 
paper \cite{LM} which focused on the special relativistic setting with the objective here being to see to what extent the conclusions drawn in the former paper carry over to 
the setting of curved spacetime. Here too the scheme proposed in \cite{LM} based on selected Orlicz spaces, proves to be essential for the analysis of QFT in curved 
spacetime. These local von Neumann algebras are constructed from local $C^*$-algebras by firstly associating a canonical naturally occurring quasi-free state to each local 
$C^*$-algebra, and then performing the GNS construction and taking the double commutant of the copy of this algebra in that representation. We hasten to add to that the use 
of quasi-free states does not stem from the fact that we want to somehow insert a Hadamard-like vacuum state into the theory, but rather from the fact that quasi-free states 
are natural objects which harmonise very well with the role played by the Orlicz space $L^{\cosh - 1}$ in noncommutative analysis \cite{ML2}. This can be seen by noting that 
this space may be viewed as the home of observables with all moments finite \cite[\S 3]{ML2}; a property which is also at the heart of the definition of quasi-free states. 
Thus, the use of such states perfectly in line with our work on QFT published in \cite{LM}.

One objection that has occasionally been raised regarding a formalism based on von Neumann algebras is the fact that von Neumann algebras seem to require an a priori Hilbert 
space which would seem to impose a random a priori restriction on the mathematical formalism modelling AQFT. We address this concern in \S \ref{S3}. On this point we note 
that the category of von Neumann algebras is equivalent to the abstractly defined category of $W^*$-algebras ($C^*$-algebras that are also Banach dual spaces). When such an 
algebra $M$ is equipped with a faithful normal semifinite weight $\nu$, the there is a unique canonical Hilbert space that may be constructed from the pairing $(M,\nu)$. 
These facts are elegantly captured by the theory of \emph{standard forms} of von Neumann algebras. We specifically show that von Neumann algebras in standard form do in fact 
form a tensorial category. In this approach where QFT is described by a net of local von Neumann algebras, there is no global Hilbert space as such. Rather each local von 
Neumann algebra describing some bounded region of spacetime comes equipped with its own unique ``canonical'' Hilbert space. Hilbert space geometry therefore turns out to be a 
local phenomenon which is induced by the local algebras themselves. At this point it is very important to note that naturally occuring Hilbert spaces at the local level opens 
the door for the incorporation into AQFT of a theory of entanglement at a local level.

In part 1 of the paper (comprising sections \ref{S4}-\ref{LieDer}) we investigate the role of (quantum) Killing vector fields in AQFT. We in particular show that in 
important examples local flows of Killing vector fields may locally be lifted to the algebra level. We also provide very mild criteria under which infinitesimal 
generators of this lifted action exist as (locally) densely defined *-derivations -- operators we may regard as quantum Lie derivatives. This therefore shows that the 
covariance property for global flows noted by Brunetti, Fredenhagen and Verch in \cite[Proposition 2.3(b)]{BFV}, will for good models extend to the local setting as well. It 
is well known that the Killing flows on Minkowski space correspond exactly to the action of the Poincar\'e group. The local covariance property of Killing flows in the 
present context may therefore also be seen as the curved analogue of the covariance property associated with the Poincar\'e group in the Minkowski space setting. The 
investigation into Killing flows therefore continues the investigation of local dynamics initiated in \cite{LM}, but now in a curved setting.

In part 2 (sections \ref{S7}-\ref{S9}) we consider local von Neumann algebras constructed from appropriate representations of Bosonic $\mathrm{CCR}$ algebras. Here our work 
is guided by three questions: Firstly an identification of those states which yield ``good'' representations of such algebras, secondly to show that some of the common 
examples of nets of of local von Neumann algebras do in fact describe a locally covariant quantum field theory in the sense of \cite[Definition 2.1]{BFV}, and thirdly to 
investigate the the extent to which the applications of Orlicz space geometry to AQFT, as demonstrated in \cite{LM} for the Minkowski space setting, can be carried over to 
the curved setting. So the fact that we here employ a construction similar to that given in Sorkin-Johnson state construction, stems from our desire to demonstrate that 
the well-established examples in the literature align with our approach.

In part 3 we finally discuss conclusions and options for further development.   

\section{Preliminaries}\label{S2}

Throughout the pair $(M,g)$ will denote a Lorentzian manifold $M$ equipped with a Lorentzian metric $g$. Such manifolds will also sometimes be referred to as spacetimes. A particularly 
class of Lorentzian manifolds are the globally hyperbolic manifolds. There are several equivalent descriptions of global hyperbolicity. Two particularly elegant conditions for global hyperbolicity 
are that the spacetime admits a Cauchy hypersurface, or alternatively that $M$ is isometric to $\mathbb{R}\times S$ with metric $-\beta dt^2+g_t$ where $\beta$ is a smooth positive function, 
$g_t$ is a Riemannian metric on $S$ depending smoothly on $t \in \mathbb{R}$ and each $\{t\}\times S$. We shall in this paper restrict attention to globally hyperbolic spacetimes. Given a 
globally hyperbolic spacetime $M$, the class of open relatively compact globally hyperbolic submanifolds $\mathcal{O}\subset M$ will be denoted by $\mathcal{K}(M,g)$. The precise content of what we 
mean by this statement is given in the discussion following definition \ref{D.causal}.

As was the case in \cite{LM}, we shall here be concerned with local algebras generated by field operators which are solutions of the Klein-Gordon equation. $C^{\infty}_0(M)$ will 
denote the space of smooth, real valued functions on $M$ which have compact support. Following Dimock \cite{Dim}, Brunetti, Fredehagen and Verch \cite{BFV} and B\"ar, Ginoux and Pf\"affle 
\cite{BGP}, we will on the manifold $M$ describe the $\mathrm{CCR}$ algebra of bosonic fields given by solutions of the Klein-Gordon equation. The global hyperbolicity of $M$ guarantees the 
existence of global fundamental solutions $\psi$ for the Klein-Gordon equation $(\Box + m^2 + \xi R)\phi = 0$, where $m\geq 0$ , $\xi \geq 0$ are constants, and $R$ is the scalar curvature of the 
metric on $M$. Given any vector bundle $E\to M$, we write $\mathscr{D}(M,E)$ for the space of compactly supported smooth sections in $E$. The solutions of the Klein-Gordon equation are of the form 
$G(f)$ ($f\in \mathscr{D}(M,E))$, where the operator $G$ is $G = G^{+} - G^{-}$ where $G^{+/-}$ are respectively the advanced/retarded Green's operators $G^{+/-}: \mathscr{D}(M,E)\to C^{\infty}(M,E)$.

When $C^*$-algebras are in view we shall denote them by $\mathcal{A}, \mathcal{B}$, with the symbols $\mathcal{M}, \mathcal{N}$ being used to denote von Neumann algebras. Given a faithful normal 
semifinite weight $\nu$ on a von Neumann algebra, one may of course pass to the GNS-representation $(\pi_\nu:\mathcal{M}\to B(H_\nu))$ of $\cM$ with respect to $\nu$. We will hereafter abbreviate 
\emph{faithful normal semifinite} to just \emph{fns}. With $\eta$ denoting the canonical embedding of the left-ideal $\mathfrak{n}_\nu=\{a\in \mathcal{M}\colon \nu(a^*a)<\infty\}$ into $H_\nu$, one 
may of course densely define an antilinear closable operator $S_0$ on $H_\nu$ by the prescription $\eta(a)\to \eta(a^*)$, $a\in \mathfrak{n}_\nu \cap \mathfrak{n}^*_\nu$. With $S$ denoting the closure 
of $S_0$, the operator $\Delta=S^*S$ is then referred to as the modular operator and the anti-unitary $J$ in the polar decomposition $S=J\Delta^{1/2}$ as modular conjugation. It is well known that the 
prescription $\sigma_t^\nu:a\to \pi_\nu^{-1}(\Delta^{it}\pi_\nu(a)\Delta^{-it})$, where $a\in \mathcal{M}$, defines an automorphism group on $\cM$ (the so called modular automorphism group) and 
$j_\nu:f\to JfJ$ ($J \in \pi_\nu(\mathcal{M})$) an anti-linear $*$-isomorphism from $\pi_\nu(\mathcal{M})$ to $\pi_\nu(\mathcal{M})'$. For the sake of ease of notation we will identify $\mathcal{M}$ with 
$\pi_\nu(\mathcal{M})$. The natural positive cone $\mathscr{P}_\nu$ of $H_\nu$ may then be defined to be the closure of the set $\{\pi_\nu(a)J_\nu(\eta(a))\colon 
a\in \mathfrak{n}_\nu\cap\mathfrak{n}^*_\nu\}$ (\cite[p. 146]{Tak}). In the case where $\nu$ is a state, this formula simplifies to $\{\eta(aj(a))\colon a\in \mathcal{M}\}$. The content of Haagerup's 
theorem regarding the standard form of a von Neumann algebra, is that the quintuple $(\mathcal{M}, \pi_\nu, H_\nu, j_\nu, \mathscr{P}_\nu)$ uniquely identifies the GNS representation of $\mathcal{M}$ 
up to a spatial $*$-isomorphism. 

Modular theory is also the gateway to the theory of quantum $L^p$ and Orlicz spaces, in that for general von Neumann algebras $\mathcal{M}$ one first needs to enlarge the algebra to the 
\emph{crossed product} $\qM=\mathcal{M}\rtimes_\nu\mathbb{R}$. This algebra admits a canonical faithful normal semifinite trace which then allows us to further enlarge the algebra to the algebra of 
$\tau$-measurable operators $\widetilde{\qM}$ affiliated with $\qM$. The actual construction of quantum $L^p$ and Orlicz spaces then happens inside this algebra; for brief details regarding this 
process we refer the reader to \cite{LM}, and to \cite{GLbook} for a detailed account. We shall here have occasion to occasionally refer to the Orlicz space $L^{\cosh-1}(\cM)$.
 
The focus of this paper is of course local algebras of Lorentzian spacetimes. These are collections of $C^*$ (alternatively von Neumann) algebras $\cA(\mathcal{O})$ each corresponding to an element 
$\mathcal{O}$ of some regular class of submanifolds in such a way the entire collection indexed by the submanifolds forms what is known as a (weak) quasi-local algebra. The basic idea here is that 
each local algebra $\cA(\mathcal{O})$ is a home for the observables of that region of spacetime. For the sake of the reader we record the definition of a quasi-local $C^*$-algebra as given in 
\cite{BGP}. The indexing set $I$ needs to be a directed set satisfying a (weak) orthogonality relation. Specifically a set $I$ is a directed set satisfying an orthogonality relation if it is a 
partially ordered set also carrying a relation $\perp$ in such a way that 
\begin{enumerate}
\item for all $\alpha, \beta\in I$ there exists some $\gamma\in I$ with $\gamma\geq \alpha$ and $\gamma\geq \beta$;
\item for every $\alpha\in I$ we can find some $\beta\in I$ with $\beta\perp\alpha$;
\item if $\alpha\leq \beta$ and $\beta\perp\gamma$ then also $\alpha\perp\gamma$;
\item if $\alpha\perp\beta$ and $\alpha\perp\gamma$, there exists $\delta\in I$ such that $\delta\geq \beta$, $\delta\geq \gamma$ and $\delta\perp\alpha$
\end{enumerate} 
If only conditions (1)-(3) are satisfied we refer to the orthogonality relation as a weak orthogonality relation. A quasi-local algebra (alt. weak quasi-local algebra) is a pair 
$(\cA, \{\cA_\alpha\}_{\alpha\in I})$ where $I$ is a directed set satisfying an orthogonality relation (alt. weak orthogonality relation) such that the following holds:
\begin{itemize}
\item[(i)] $\cA_\alpha\subset\cA_\beta$ whenever $\alpha\leq \beta$;
\item[(ii)] the algebras $\cA_\alpha$ all carry a common unit;
\item[(iii)] $\cup_\alpha\cA_\alpha$ is norm-dense in $\cA$;
\item[(iv)] if $\alpha\perp \beta$ then $[\cA_\alpha, \cA_\beta]=\{0\}$.
\end{itemize} 
For a collection of local algebras the partial order on the indexing collection of submanifolds is containment, with $\mathcal{O}_1$ said to be orthogonal to $\mathcal{O}_2$ if the regions 
$\mathcal{O}_1$ and $\mathcal{O}_2$ are causally independent. To obtain the corresponding notion of a von Neumann quasi-local algebra one simply replaces $\cA$ and the $\cA_\alpha$'s with 
von Neumann algebras and replaces the norm-density in (iii) with $\sigma$-weak density. 

The local algebra generated by solutions to the Klein-Gordon equation and indexed by $\mathcal{K}(M,g)$ does turn out to be a quasi-local algebra. Following \cite{BGP} we may construct this algebra 
as follows: For each $\mathcal{O} \in \mathcal{K}(M,g)$, one firstly constructs $L^2(\mathcal{O})$ where integration is by the volume form on $\mathcal{O}$. The algebra $\cA(\mathcal{O})$ is then a 
copy of the $\mathrm{CCR}$ algebra constructed from Weyl operators $W(\psi)$ where the $\psi$ ranges over the solutions of the local Klein-Gordon equation on $\mathcal{O}$, namely 
$\{\psi\colon \psi=\widetilde{G}_{\mathcal{O}}(f), f\in \mathscr{D}(\mathcal{O},E)\}$. Here we note that the operator $\widetilde{G}_{\mathcal{O}}$ is for every $f\in \mathscr{D}(\mathcal{O},E)\}$ 
given by $\widetilde{G}_{\mathcal{O}}(f)=G(f_{ext}){\upharpoonright}{\mathcal{O}}$ where $f_{ext}$ denotes the element of $\mathscr{D}(M,E)$ obtained by assigning the value 0 to all points 
$p\in M\backslash\mathcal{O}$. The local algebras of Brunetti, Fredenhagen and Verch \cite{BFV} are very similar with the main difference being that their Weyl operators are constructed from 
elements of the set $\{\psi\colon \psi=\widetilde{G}_{\mathcal{O}}(f), f\in C_0^\infty(\mathcal{O})$. Readers wishing to see an axiomatic description of this approach would be well-advised to 
consult the papers of Fredenhagen and Rejzner \cite{FR}, and Buchholz and Fredenhagen \cite{BF}.

Von Neumann local algebras are typically constructed by passing to a representation by means of a regular state of the $C^*$-algebra setting and then taking the double commutant of the 
representation. At the forefront of such ``good'' states are the quasi-free states. Representations with these states typically lead to the iconical Araki-Woods factors. See \cite{Der} 
for an account of this process. However in many cases more than just quasi-freeness is needed. Hadamard states have the added advantage that they are related to Quantum Weak Energy 
Inequalities (QWEIs) which in turn ensure that the existence of ``good'' Wightman fields at the microscopic level harmonises with observance of the second law of thermodynamics at the 
macroscopic level. Further physical reasons also bring passive and stationary states into the picture. This topic will be explored in more detail in Part \ref{Pt.2}. As noted previously 
\cite{Der} is a good reference for information on quasi-free states and Araki-Woods factors. The paper \cite{KM} gives useful information regarding quasi-free, Hadamard and stationary states.

Besides the references mentioned above we for the sake of the reader note the following:  The book of Wald \cite{Wald2} contains valuable information regarding relativity and differential 
geometry. Taking the appendices into account, the book \cite{BGP} will provide the reader with a very good ``quick'' introduction to Lorentzian geometry and an insightful account of the 
Klein-Gordon and related equations. The books of Gallier and Quiantance \cite{GQ1, GQ2} give an excellent account of local flows albeit in a Riemann geometric context. Minguzzi's notes 
\cite{Min} give a very comprehensive account of causality theory. 

\section{The standard form of von Neumann algebras}\label{S3}

We will show that there is a categorical way to incorporate Hilbert spaces into the present framework without going to the extreme of assuming the a priori 
existence of some global Hilbert space of states. On a philosophical level we agree with the fact that the theory should be based on the principles of AQFT. In 
the modern approach to canonical quantization this means that to each quantum system there corresponds a $C^*$ or von Neumann algebra which encodes the 
information of that system. For us the basic object is however an abstract von Neumann algebra equipped with some fns weight. With this approach each such 
abstract von Neumann algebra describing observables associated with some local region of spacetime allows for the construction of a unique Hilbert space associated 
with that particular algebra by means of the GNS process. This Hilbert space is therefore a product of the system rather than an a priori given object. 

For the reader's convenience, we recall that the equivalent abstract definition of von Neumann algebras identifies these algebras as $C^*$-algebra which have a 
(necessarily unique) predual. In fact such abstract von Neumann algebras are the departure point for the treatment of this topic in both of \cite{Sakai} and 
\cite{Tak1}. What is crucial for our work is the following well known fact: any von Neumann algebra $\mathcal{M}$ (abstract or concrete) admits a faithful normal 
weight $\nu$ for which the associated GNS construction yields a faithful standard representation $\pi$ of $\mathcal{M}$, i.e. a normal $^*$-isomorphism into the 
algebra of all bounded linear operators $B(H_\nu)$ on the GNS Hibert space $H_\nu$ with the image $\pi(\mathcal{M})$ being a von Neumann algebra in standard form, 
see either \cite[Theorem IX.1.2 \& Definition IX.1.13]{Tak} or \cite[Theorem 10.59]{GLbook}.

At this point there are two related categorical options for taking the theory further. We know from Maharam's theorem that the category of abelian von Neumann algebras equipped with some 
faithful normal semifinite trace is equivalent to the category of localizable measure algebras (see \cite[Chapter 1]{GLbook}). In a quantum setting the pairings $(\mathcal{M},\nu)$ 
where $\nu$ is a faithful normal semifinite weight may therefore be regarded as noncommutative analogues of localizable measure algebras, and which we may therefore refer to as \emph{von Neumann 
measure algebras} (cf. Definition \ref{vnmasf-cat}). It is this very fact that underlies the rationale behind using von Neumann algebras as a formalism for Quantum information Theory. Here two such pairings, $(\mathcal{M}_0,\nu_0)$ 
and $(\mathcal{M}_1,\nu_1)$, are equivalent if there is a *-isomorphism $\iota$ from $\mathcal{M}_0$ onto $\mathcal{M}_1$ such that $\nu_0=\nu_1\circ\iota$. A measure algebraic approach would 
therefore be based on a categorical analysis of these pairings. However, as we shall shortly see, as far as quantum information theory is concerned an approach which more strongly emphasises the 
GNS construction may be more beneficial as those are the tools which facilitate a description of entanglement. The two approaches are interlinked in that equivalent pairings 
$(\mathcal{M},\nu)$ will produce standard forms that are copies of each other, but they are not synonymous. When bringing the theory of the standard form of a von Neumann algebra into 
play, it quickly becomes clear von Neumann algebras in standard form are also a category which admit a tensor structure and which does not require an a priori global Hilbert space, since 
the form itself encodes the ``canonical'' Hilbert space for that von Neumann algebra. On a related note, we observe that since the von Neumann algebra tensor product is unique there is, 
in contrast to the $C^*$ context, here no selection of the ``correct'' tensor product to be made. In addition the tensor structure of standard forms of von Neumann algebras is refined 
enough to allow for a description of entangled states, and therefore by extension a refined theory of quantum information theory. We briefly outline these facts.

The standard form of a von Neumann algebra is an abstract characterisation of the uniqueness of the GNS-representation of the pair $(\cM,\nu)$ where $\nu$ is a faithful normal 
semifinite weight on the von Neumann algebra $\cM.$ Just as in classical measure where replacing a given measure with an equivalent measure will leave most of the measure theoretic 
constructs intact, different faithful normal semifinite weights may produce a GNS space which in all aspects is a copy of the first. Haagerup proved a~very deep theorem giving 
expression to this fact by essentially showing that any representation of $\cM$ which admits objects that mimic the action of $J_\nu$ and $\mathscr{P}_\nu^\natural$ is a faithful 
copy of the GNS-representation of the pair $(\cM,\nu).$ This claim may be made exact 
with the following definition:

\begin{definition}\label{D8.defstdfrm} Given a von Neumann algebra $\cM$ equipped with a~faithful normal semifinite weight $\nu$, a quintuple 
$(\cM, \pi_0, H_0, J, \mathscr{P})$ where $\pi_0$ is a faithful representation of $\cM$ on the Hilbert space $H_0$, $J:H_0\to H_0$ anti-linear isometric 
involution, and $\mathscr{P}$ a self-dual cone of $H_0$, is said to be a \emph{standard form} of $\cM$ if the following conditions hold:
\begin{itemize}
\item $J\pi_0(\cM) J=\pi_0(\cM)'$ (the commutant of $\pi_0(\cM)$),
\item $JzJ=z^*$ for all $z$ in the centre of $\pi_0(\cM)$,
\item $J\xi=\xi$ for all $\xi\in \mathscr{P}$,
\item $a(JaJ)\mathscr{P}\subseteq \mathscr{P}$ for all $a\in \pi_0(\cM)$.
\end{itemize}
(Recall that when we say that $\mathscr{P}$ is a self-dual cone, we mean that $\xi\in \mathscr{P}$ if and only if $\langle\xi,\zeta\rangle\geq 0$ for all 
$\zeta\in \mathscr{P}.$)
\end{definition}

The value of the above concept is derived from the following very deep and useful theorem:

\begin{theorem}[{\cite{haag-stdfm}}]\label{T8.haag-stdfm}
The standard form of a von Neumann algebra $\cM$ is unique in the sense that if $$(\cM, \pi, H, J, \mathscr{P})\mbox{ and }(\widetilde{\cM}, \widetilde{\pi}, \widetilde{H}, 
\widetilde{J}, \widetilde{\mathscr{P}})$$are two standard forms, and $\alpha:\pi(\cM)\to\widetilde{\pi}(\widetilde{\cM})$ is a $*$-isomorphism, 
then there exists a unique unitary operator $u:H\to\widetilde{H}$ such that
\begin{itemize}
\item $\alpha(x)=uxu^*$ for $x\in \pi_0(\cM)$;
\item $\widetilde{J}=uJu^*$;
\item $\widetilde{\mathscr{P}}=u\mathscr{P}.$
\end{itemize}
\end{theorem}

\begin{definition}\label{Stdfm-cat}
On the basis of the above theorem we define a category \textbf{VN-SF} for which the objects are the quintuples described in Definition \ref{D8.defstdfrm}. 
The isomorphisms for this category are the spatial *-isomorphisms described in Theorem \ref{T8.haag-stdfm}. 

Next let $(\mathcal{N},\nu_{\mathcal{N}})$ and $(\mathcal{M},\nu_{\mathcal{M}})$ be pairings for which $\mathcal{N}$ is for some projection $e\in \cM$ 
a sub-von Neumann algebra of $e\cM e$ and $\nu_{\cN}=\nu_{\mathcal{M}}{\upharpoonright}{\cN}$ to $\cN$ is still semifinite. Let $e_{(\cM:\cN)}$ be the projection of $H_{\nu_{\cM}}$ 
(the GNS space of the pairing $(\mathcal{M},\nu_{\mathcal{M}})$) onto the subspace $H_{\cN}$ generated by $\{\eta_\nu(a):a\in \mathfrak{n}_{\nu_{\cN}}\cap\mathfrak{n}^*_{\nu_{\cN}}\}$. It can 
easily be seen that $H_{\cN}$ is a copy of $H_{\nu_{\cN}}$. We then say that the standard form of a pairing $(\mathcal{N},\nu_{\mathcal{N}})$ embeds into the standard form of the pairing 
$(\mathcal{M},\nu_{\mathcal{M}})$ if it is of the form $(\cN, \pi_{\nu_{\cM}}{\upharpoonright}{\cN}, H_\cN, J_{\nu_{\cM}}{\upharpoonright}{H_{\cN}}, \mathscr{P}_{\nu_{\cN}})$ where 
$\mathscr{P}_{\nu_{\cN}}\subset \mathscr{P}_{\nu_{\cM}}$.

Morphisms in the category \textbf{VN-SF} are the simply defined to be any combination of embeddings and isomorphisms.
\end{definition}

\begin{remark} It is clear from \cite[\S 10.5]{GLbook} that subalgebras $\mathcal{N}$ of some von Neumann algebra $\mathcal{M}$ equipped with a faithful normal semifinite weight $\nu$ 
for which there either exists a normal conditional expectation $\mathbb{E}:\mathcal{M}\to \mathcal{N}$ with $\nu\circ\mathbb{E}=\nu$, or a projection $e$ in the centralizer 
$\mathcal{M}_\nu$ for which $\mathcal{N}=e\mathcal{M}e$ with $\nu{\upharpoonright}{e\mathcal{M}e}$ still semifinite, provide particular examples of the standard form embeddings described above.
\end{remark}

The category that best suits our needs is in fact the subcategory of the above for which the objects are standard forms of  pairings $(\cM,\nu)$ of von Neumann algebras and an associated fns weight (referred to as 
von Neumann measure algebras in standard form), since this subcategory retains the benefits of access to standard forms, but more strongly emphasises the integration theory aspect. 

\begin{definition}\label{vnmasf-cat}
The category \textbf{VN-MA-SF} is defined as follows:
\begin{description}
\item[Objects] The objects consist of quintuples $(\cM, \nu, H_\nu, J_\nu, \mathscr{P}_\nu)$ where $\cM$ is a von Neumann algebra equipped with a faithful normal semifinite weight, and 
$(H_\nu, J_\nu, \mathscr{P}_\nu)$ respectively the Hilbert space, modular conjugation and natural positive cone generated by performing the GNS construction for the pair $(\cM, \nu)$. 
\item[Morphisms] Two objects $(\cM_1, \nu_1, H_{\nu_1}, J_{\nu_1}, \mathscr{P}_{\nu_1})$ and $(\cM_2, \nu_2, H_{\nu_2}, J_{\nu_2}, \mathscr{P}_{\nu_2})$ are said to be isomorphic 
if there exists a *-isomorphism $\alpha$ from $\cM_1$ onto $\cM_2$ such that $\nu_2\circ\alpha = \nu_1$. (Recall that we have already noted that if such a *-isomorphism exists, 
it will induce a unitary operator from $H_{\nu_1}$ to $H_{\nu_2}$ which identifies the triples $(H_{\nu_1}, J_{\nu_1}, \mathscr{P}_{\nu_1})$ and $(H_{\nu_2}, J_{\nu_2}, \mathscr{P}_{\nu_2})$.) 
We say that $(\cM_1, \nu_1, H_{\nu_1}, J_{\nu_1}, \mathscr{P}_{\nu_1})$ embeds into $(\cM_2, \nu_2, H_{\nu_2}, J_{\nu_2}, \mathscr{P}_{\nu_2})$ if for some projection $e\in \cM_2$, $\cM_1$ 
is a subalgebra of $e\cM_2 e$, $\nu_1=\nu_2{\upharpoonright}\cM_1$, $H_{\nu_1}\subseteq H_{\nu_2}$, $J_{\nu_1}=J_{\nu_2}{\upharpoonright}H_{\nu_1}$ and 
$\mathscr{P}_{\nu_1}\subseteq \mathscr{P}_{\nu_2}$. Morphisms are then defined to be combinations of isomorphisms and embeddings.
\end{description}
\end{definition}

The following Lemma provides criteria under which the condition $J_{\nu_1}=J_{\nu_2}{\upharpoonright}H_{\nu_1}$ pertains.

\begin{lemma}\label{mod-subalg2} Let $\cM$ be a von Neumann algebra equipped with an fns weight $\nu$ and let $\cN$ be a von Neumann subalgebra for which $\nu{\upharpoonright}\cN$ is still semifinite. 
Then the GNS Hilbert space of the pair $(\cN, \nu{\upharpoonright}\cN)$ corresponds to the subspace $H_{\cN}=\overline{\{\eta_{\cM}(a):a\in \mathfrak{n}_\nu\cap\mathfrak{n}^*_\nu\cap\cN\}}$ of 
$H_\nu=H_{\cM}$. If $H_{\cN}$ is an invariant subspace of the modular conjugation operator $J_{\cM}$ for $(\cM, \nu)$, then under this identification $J_{\cM}{\upharpoonright}H_{\cN}$ is exactly the 
modular conjugation operator $J_{\cN}$ for the pair $(\cN, \nu{\upharpoonright}\cN)$, in which case $\mathscr{P}_{\cN}\subseteq\mathscr{P}_{\cM}$
\end{lemma}

\begin{proof} It is clear that the closable densely defined anti-linear operator $$S_0^{\cN}:\eta_{\cM}(a)\mapsto\eta_{\cM}(a^*)\mbox{ for all }a\in \mathfrak{n}_\nu\cap\mathfrak{n}^*_\nu\cap\cN$$ 
corresponding to the pair $(\cN, \nu{\upharpoonright}\cN)$, is a restriction of the closable densely defined anti-linear operator 
$$S_0^{\cM}:\eta_{\cM}(a)\mapsto\eta_{\cM}(a^*)\mbox{ for all }a\in \mathfrak{n}_\nu\cap\mathfrak{n}^*_\nu.$$  
If momentarily we regard the closed extension $S^{\cN}$ of $S_0^{\cN}$ as a map into $H_{\cM}$, it is then clear from \cite[Lemma VI.1.4]{Tak} that $S^{\cM}$ is an extension of $S^{\cN}$. That is 
$S^{\cM}{\upharpoonright}\mathrm{dom}(S^{\cN})=S^{\cN}$. Since $H_{\cN}$ is an invariant subspace of $J_{\cM}$, the operator $J_{\cM}$ restricts to an anti-unitary on $H_{\cN}$. Let us write 
$\widehat{J}_{\cM}$ for this restriction. We will then have that $\widehat{J}_{\cM}S^{\cN}$ is a closed operator on $H_{\cN}$ for which we have that 
$$\widehat{J}_{\cM}S^{\cN}= J_{\cM}S^{\cM}{\upharpoonright}\mathrm{dom}(S^{\cN})= \Delta^{1/2}_{\cM}{\upharpoonright}\mathrm{dom}(S^{\cN}).$$ In view of the fact that $\mathrm{dom}(S^{\cN}) 
\subset \mathrm{dom}(S^{\cM})$, it is palpably clear that $\langle \xi, \Delta^{1/2}_{\cM}\xi\rangle\geq 0$ for all $\xi \in \mathrm{dom}(S^{\cN})$ and hence that 
$\Delta^{1/2}_{\cM}{\upharpoonright}\mathrm{dom}(S^{\cN})$ is a positive operator. By construction $S^{\cN}= \widehat{J}_{\cM}(\Delta^{1/2}_{\cM}{\upharpoonright}\mathrm{dom}(S^{\cN}))$ is then 
a polar decomposition of $S^{\cN}$. Since $\Delta^{1/2}_{\cN}$ has dense range, the uniqueness of the polar decomposition now ensures that $J_{\cN}=\widehat{J}_{\cM}$ as was required. (Although we shall not 
need this, we also have that $\Delta^{1/2}_{\cN}= \Delta^{1/2}_{\cM}{\upharpoonright}\mathrm{dom}(S^{\cN})$.) 

The final claim is now an obvious consequence of the definition of the natural positive cone \cite[Definition 8.40]{GLbook}. 
\end{proof}

\subsection{\textbf{VN-SF} and \textbf{VN-MA-SF} are tensorial categories}

As can be seen from the following result, the von Neumann algebra tensor product is a tensor structure on the category \textbf{VN-SF}. 

\begin{proposition}[Proposition 2a, \cite{IC}]\label{vnstdform_1} Let $(\cM, \pi_{\cM}, H_{\cM}, J_{\cM}, \mathscr{P}^+_{\cM})$ and 
$(\cN, \pi_{\cN}, H_{\cN}, J_{\cN}, \mathscr{P}^+_{\cN})$ be standard forms for $\cM$ and $\cN$. Then $(\cM\overline{\otimes}\cN, \pi_{\cM}\otimes \pi_{\cN}, 
H_{\cM}\overline{\otimes}H_{\cN}, J_{\cM}\otimes J_{\cN}, \mathscr{P}^+)$ is a standard form for $\cM\overline{\otimes}\cN$ where $\mathscr{P}^+$ is the closure 
of $\{aJaJ(x\otimes y)\colon a\in (\cM\overline{\otimes}\cN), x\in \mathscr{P}^+_{\cM}, y\in\mathscr{P}^+_{\cN}\}$.
\end{proposition} 

The above result even holds for infinite tensor products (see \cite[Proposition 2b]{IC}). This tensor structure also behaves well with regard to the subcategory \textbf{VN-MA-SF}. 

\begin{proposition}[Proposition 3g, \cite{IC}] Let $(\cM, \pi_{\cM}, H_{\cM}, J_{\cM}, \mathscr{P}^+_{\cM})$ and $(\cN, \pi_{\cN}, H_{\cN}, J_{\cN}, \mathscr{P}^+_{\cN})$ 
be standard forms for $\cM$ and $\cN$ with these algebras are respectively equipped with the fns weights $\nu_{\cM}$ and $\nu_{\cN}$.

Let $(\cM\overline{\otimes}\cN, \pi_{\cM}\otimes \pi_{\cN}, H_{\cM}\overline{\otimes}H_{\cN}, J_{\cM}\otimes J_{\cN}, \mathscr{P}^+)$ be the standard form for 
$\cM\overline{\otimes}\cN$ described in Proposition \ref{vnstdform_1}. Then the following holds:
\begin{itemize} 
\item There exists a weight $\nu_{\cM}\otimes\nu_{\cN}$ such that $a\in \mathfrak{n}_{\nu_{\cM}}$ and $b\in \mathfrak{n}_{\nu_{\cN}}$ imply 
$a\otimes b\in \mathfrak{n}_{\nu_{\cM}\otimes\nu_{\cN}}$ with $\eta_{\nu_{\cM}\otimes\nu_{\cN}}(a\otimes b)= \eta_{\nu_{\cM}}(a)\otimes \eta_{\nu_{\cN}}(b)$.
\item $\Delta_{\nu_{\cM}\otimes\nu_{\cN}}=\Delta_{\nu_{\cM}}\otimes\Delta_{\nu_{\cN}}$ (equivalently 
$S_{\nu_{\cM}\otimes\nu_{\cN}}=S_{\nu_{\cM}}\otimes S_{\nu_{\cN}}$).
\item $(\nu_{\cM}\otimes\nu_{\cN})'=(\nu_{\cM})'\otimes(\nu_{\cN})'$ where $\nu'$ stands for the ``opposite weight'' on the commutant \cite[Definition VII.1.18]{Tak}.
\item Let $A\subset\mathfrak{n}_{\nu_{\cM}}$ and $B\subset\mathfrak{n}_{\nu_{\cN}}$ be linear subspaces which are self-adjoint in the sense that 
$A=\{a^*\colon a\in A\}$ and $B=\{b^*\colon b\in B\}$. If $\eta_{\nu_{\cN}}(A)$ and $\eta_{\nu_{\cN}}(B)$ are respectively cores for $S_{\nu_{\cM}}$ and 
$S_{\nu_{\cN}}$, then $\nu_{\cM}\otimes\nu_{\cN}$ is the unique fns weight for which $a\otimes b\in \mathfrak{n}_{\nu_{\cM}\otimes\nu_{\cN}}$ for all $a\in A$ 
and $b\in B$.
\end{itemize}
\end{proposition}

\subsection{The structure of both \textbf{VN-SF} and \textbf{VN-MA-SF} allow for entanglement}\label{ss-entangle} 

In QFT, the Reeh-Schlieder theorem implies that any chosen vector $\Psi$ can be approximated by the action of some localized operator on the vacuum. This result shows that vacuum posses small 
but nonvanishing and even long distance correlations which are not expected in classical Physics. So quantum correlations appeared, c.f. the Haag book \cite{Haag} for details. 

On the other hand AQFT, as it was mentioned, puts the emphasis on localized observables. We remind that in Quantum Theory knowledge of a system is described by a 
state. Thus quantum correlations should be encoded in some specific states. In particular, there are states of two (or a group) localized particles, even when the 
particles are separated by a large distance, which are not just a convex combination of individual quantum states. Such states are called entangled states and 
what is important, they describe quantum (so non-classsical) correlations. We remind that a state which is a convex combination of individual quantum states is 
called a separable state. It is worth pointing out that such states contain only classical correlations.

Consequently, in any axiomatic description of quantum fields, so even in Quantum Gravity, a characterization of quantum entanglement should be given. In this 
section we present within the framework of $W^*$-algebraic QFT, a general description of entangled states. This is a generalization of the characterization of 
entangled and separable states given in Dirac's approach to Quantum Mechanics \cite{Maj}.

We begin by recalling some more facts of the theory of von Neumann algebra in the standard form. Denote by ${\cM}_*$ the predual of von Neumann algebra $\cM$. 
If $\cM$ is $\sigma$-finite, then to each $\omega \in {\cM}_*^+$, there corresponds a unique $\xi \in \cP$ with 
$\omega(\cdot)  = (\xi, \cdot \xi) \equiv \omega_{\xi}(\cdot)$. (See \cite{Araki} and \cite{BR}, \cite{Tak} for a recent account of the theory.) Furthermore, the mapping $\xi \mapsto \omega_{\xi}$ 
is a homeomorphism when both $\cP$ and ${\cM}^+_*$ are equiped with the norm topology.

Let us consider a composite system associated with two regions, say $1 + 2$; let $(\cM, \pi_{\cM}, H_{\cM}, J_{\cM}, \mathscr{P}^+_{\cM})$ be associated with the 
region $1$ ( $(\cN, \pi_{\cN}, H_{\cN}, J_{\cN}, \mathscr{P}^+_{\cN})$ with the region $2$ respectively). The structure of the von Neumann algebra associated with 
the region $1+2$ is described by Proposition \ref{vnstdform_1}. It is of crucial importance that the set of normal states of the composite system is not a convex 
combination of tensor products of individual normal states. More precisely:

\begin{proposition}[Prop 2c, \cite{IC}]\label{vnstdform_2} Let $(\cM, \pi_{\cM}, H_{\cM}, J_{\cM}, \mathscr{P}^+_{\cM})$ and $(\cN, \pi_{\cN}, H_{\cN}, J_{\cN}, \mathscr{P}^+_{\cN})$ be standard forms 
for $\cM$ and $\cN$ and $(\cM\overline{\otimes}\cN, \pi_{\cM}\otimes \pi_{\cN}, H_{\cM}\overline{\otimes}H_{\cN}, J_{\cM}\otimes J_{\cN}, \mathscr{P}^+)$ the standard form for 
$\cM\overline{\otimes}\cN$ described in Proposition \ref{vnstdform_1}. Let $\mathscr{P}^+_0$ be the 
closed convex hull of $\{(x\otimes y)\colon x\in \mathscr{P}^+_{\cM}, y\in\mathscr{P}^+_{\cN}\}$. Then the following holds: 
\begin{itemize} 
\item If either of $\cM$ or $\cN$ is commutative, then $\mathscr{P}^+ = \mathscr{P}^+_0$.
\item If both $\cM$ and $\cN$ are noncommutative, then $\mathscr{P}^+ \neq \mathscr{P}^+_0$.
\end{itemize}
\end{proposition} 

The striking conclusion for physics is: separable normal states of composition system, describing classical correlations, form only a subset in the set of normal 
states. Consequently, there are states $\mathscr{P}^+ \setminus \mathscr{P}^+_0$, entangled states, which encode non-classical - quantum correlations. Keeping the 
Reeh-Schlieder theorem in mind, this result is not surprising. A system composed of two quantum subsystems must have quantum correlations. However, the great 
advantage of the above result is the description of states with quantum correlations.

It is worth pointing out that the above nice result was obtained in the framework of $W^*$-algebras. This is huge advantage over the $C^*$-approach to QFT.

\part{Quantum Killing vector fields for AQFT}

In the brief comment near the end of the introduction regarding $L^{\cosh-1}$-regularity we already indicated the utility of quantum integration theory for QFT, albeit 
for the special relativistic setting. We shall revisit this issue for general spacetimes in the next part. For now we pass to exploring differential structures for curved 
spacetimes. In so doing we shall in this part primarily focus on the $\mathrm{CCR}$ local algebras described in \cite{BGP}. To be more specific about differentiation, 
more specific algebras and states need to be considered. For now we content ourselves with sketching the broad panorama that begs further analysis.

\section{Killing local flows on compacta of Lorentzian manifolds}\label{S4}

It is known that the local algebraic formalism as proposed in \cite{BFV} behaves well with respect to groups of isometries (see \cite[Proposition 2.3(b)]{BFV}). However Killing 
vector fields generally correspond to \emph{local} flows of isometries. So to fully understand how concepts like Killing vector fields may be lifted to the algebraic context, 
we need to see to what extent the content of \cite[Proposition 2.3(b)]{BFV} may be recovered when dealing with local flows of isometries. The primary challenge we need to overcome 
is the fact that at the manifold level we are now dealing with local rather than global isometries which only exhibit ``group-like'' behaviour for short times. This necessitates 
the creation of new proof strategies.

\begin{definition}\label{D.causal} Let $(M,g)$ be a time-oriented Lorentzian manifold. A tangent vector $v$ is \emph{causal} if $v\neq 0$ and $g(v,v)\leq 0$. A curve is called \emph{causal} if 
all of its tangent vectors are causal. The causal curves between two points in a set $V$ is usually denoted by $J^+(V)\cap J^-(V)$ where 
$J^+(V)=\cup_{p\in V}J^+(p)$ and $J^+(p)$ is the \emph{causal future} (resp. $J^-(p)$ the \emph{causal past}) of the point $p$, etc.  

A subset $V\subset M$ is said to be \emph{causally convex} if it contains all causal curves between two points in the set, that is if $J^+(V)\cap J^-(V)\subset V$. For 
an arbitrary subset $V\subset M$ the \emph{causally convex hull} of $V$ is defined to be the union of the set of the images of all causal curves which start and end in $V$. We shall however write 
$\mathrm{cau}(V)$ for the causal hull.
 
A subset $\Omega$ of $(M,g)$ is called \emph{causally compatible} if for all points $x \in \Omega$ we have $J^\pm_\Omega(x) = J_M^\pm(x)\cap \Omega$.
\end{definition}

Note that the inclusion $J^\pm_\Omega(x) \subseteq J_M^\pm(x)\cap \Omega$ always holds. The condition of being causally compatible therefore means that whenever 
two points in $\Omega$ can be joined by a causal curve in $M$, this can also be done inside $\Omega$.

\medskip

\textbf{Notation:} We will denote the class of all relatively compact causally compatible globally hyperbolic open subsets of a globally hyperbolic spacetime 
$(M,g)$ by $\mathcal{I}_0(M,g)$. B\"ar, Ginoux and Pf\"affle \cite{BGP} use the notation $\mathcal{I}(M,g)$ for $\mathcal{I}_0(M,g)\cup\{\emptyset, M\}$ added. 
Both these classes are clearly subclasses of the class of all open relatively compact causally convex subsets of $M$. We follow \cite{BFV} by denoting the set of 
all open relatively compact causally convex subsets of $M$ by $\mathcal{K}(M, g)$. Elements of this class can all be shown to be globally hyperbolic in their 
own right

It is of independent interest to note that global hyperbolicity may be characterised in terms of compactness of the causally convex hull,

\begin{proposition}The following are equivalent for a Lorentzian manifold: 
\begin{itemize} 
\item The manifold is globally hyperbolic.
\item The manifold is causal and the causally convex hull of any compact subset is compact. \cite[Proposition 2.5]{HoMin}
\end{itemize}
Both imply that the causally convex hull of any relatively compact set is relatively compact. \cite[Page 4]{Min} 
\end{proposition}

We are particularly concerned with Killing vector fields on Lorentzian manifolds. These are (smooth) vector fields $X$ for which the Lie derivative of the metric 
tensor vanishes, that is $L_Xg = 0$. The primary reason for our interest stems from the following fact

 \begin{remark}
It is well-known that the Killing vector fields on Minkowski space correspond exactly to the action of the Poincar\'e group on Minkowski space. So for Lorentzian 
manifolds these fields may be viewed as a curved analogue of the action of the Poincar\'e group on Minkowski space.
\end{remark}
 
Killing vector fields are very beautifully characterised as those vector fields with isometric local flows. For the sake of the reader we record the essentials.

\begin{definition}\label{Lor-isom} Let $M$ and $N$ be Lorentzian manifolds with metric tensors $g_1$ and $g_2$. An isometry from $M$ to $N$ is a 
diffeomorphism $\Psi: M \to N$ that preserves the metric tensor, that is $\Psi^*(g_2) = g_1$.   Explicitly, $\langle d\Psi_p(v), d\Psi_p(w)\rangle_{\Psi(p)} = \langle v, w\rangle_p$ 
for all $v,w\in T_p(M)$. For Lorentzian manifolds this means that for each $p$ in the domain of $\Psi$, the map $d\Psi_p$ is a Lorentzian 
isometry from $T_p(M)$ to $T_{\Psi(p)}(M)$. (See the comment after Definition 3.6 of \cite{ONeil}.) \end{definition}

\begin{theorem} 
A vector field $Z$ on a semi-Riemannian manifold is a Killing vector field iff the stages $\Psi_s^Z$ of the local flow are isometries. 
(\cite[9.21]{ONeil})
\end{theorem} 

\textbf{Note:} Lemma 4.4.8 of \cite{BGP} is very important. It shows that for globally hyperbolic spacetimes the relatively compact causally compatible globally 
hyperbolic open subsets is a directed set with an orthogonality relation. However more is true. What is actually proved is the following:

\begin{proposition}\label{directed} Let $(M,g)$ be a globally hyperbolic Lorentzian manifold. For any compact subset $K$ there exists some 
$\mathcal{O}\in \mathcal{K}(M,g)$ containing $K$. Moreover the class of subsets $\mathcal{K}(M,g)$ is a directed set with orthogonality relation. (Compare 
\cite[Lemma A.5.11]{BGP}.) Finally for any $\mathcal{O}_1\in \mathcal{I}_0(M,g)$ there exists $\mathcal{O}_2\in \mathcal{K}(M,g)$ such that 
$\mathcal{O}_1\subset\mathcal{O}_2$.
\end{proposition}

\begin{proof}  The proof of \cite[Proposition A.5.13]{BGP} relies on Proposition A.5.12 of \cite{BGP}. However the hypothesis of \cite[Proposition A.5.12]{BGP} 
understates what is actually proven. In the second sentence of that proof the authors actually show that the set constructed there is not just causally 
compatible, but in fact causally convex. So the set constructed in \cite[Proposition A.5.13]{BGP} must also be causally convex since it is constructed by an 
application of \cite[Proposition A.5.12]{BGP}. The first claim therefore follows from \cite[Proposition A.5.13]{BGP}. The second claim follows from 
\cite[Lemma 4.4.8]{BGP}. To see this note that the proof of \cite[Lemma 4.4.8]{BGP} relies on an application of \cite[A.5.11, A.5.12 \& A.5.13]{BGP}. We have 
already seen that \cite[A.5.12 \& A.5.13]{BGP} actually yield elements of $\mathcal{K}(M,g)$. A consideration of the proof of \cite[Lemma A.5.11]{BGP} shows that 
here too the authors actually prove causal convexity and not just causal compatibility. Hence the claim follows. The final claim follows by applying the first 
assertion to the fact that $\overline{\mathcal{O}}_1$ is compact.
\end{proof}

\begin{corollary}\label{compare} Let $\cA(M)$ be a local algebra of a globally hyperbolic spacetime $(M,g)$ defined as in Lemma 4.4.8 and Corollary 4.4.12 of 
\cite{BGP}. Then $\cup_{\mathcal{K}(M,g)}\cA(\mathcal{O}) = \cup_{\mathcal{I}_0(M,g)}\cA(\mathcal{O})$. 
\end{corollary}

\begin{proof} This is a straightforward consequence of the final claim of the proposition and the fact that $\mathcal{K}(M,g)\subset \mathcal{I}_0(M,g)$.
\end{proof}

\begin{corollary} Let $Z$ be an arbitrary smooth vector field and $K$ a compact neighbourhood. Then there exists some $t_{K}>0$ such that $K$ is in the domain of 
each $\Psi^Z_s$ with $0<|s|\leq t_{K}$. Given any relatively compact causally compatible globally hyperbolic open subset $\mathcal{O}$ there exists some 
$t_{\mathcal{O}}>0$ such that $\mathcal{O}$ is in the domain of each $\Psi^Z_s$ with $0<|s|\leq t_{\mathcal{O}}$.
\end{corollary} 

\begin{proof} The first part follows from a compactness argument. First consider positive indices $t_p$. Each $p\in K$ is an element of 
$\mathrm{dom}(\Psi_{t_p}^Z)$ for $t_p>0$ small enough. By compactness $K$ can be covered by finitely many of the $\mathrm{dom}(\Psi_{t_p}^Z)$'s. Since the domains 
increase as $t$ decreases to 0, taking the minimum of the remaining $t$'s yields a $t>0$ for which $K\subset \mathrm{dom}(\Psi_{s}^Z)$ for each $0<s<t$. Now repeat 
the argument for negative $t$'s. For the second claim simply apply the above to the closure of $\mathcal{O}$.
\end{proof}

\begin{lemma} 
Let $(M,g)$ be a time-oriented Lorentzian manifold and $Z$ a Killing vector field with $\Psi_s^Z$ the stages of the local flow. Then $\Psi_s^Z$ will map causal 
curves between two points in its domain to causal curves. 
\end{lemma}

\begin{proof}
Let $g$ be the canonical Lorentzian-Riemannian metric tensor on $M$ and $Z$ a Killing field with stages $\Psi_s^Z$. For part 1 all but the claim about 
preservation of causal compatibility are a consequence of continuity. Given a causally compatible subset $K$ of the domain of $\Psi_s^Z$, the claim about causal 
compatibility will follow if we can show that $\Psi_s^Z$ (and therefore also $\Psi_{-s}^Z=(\Psi_s^Z)^{-1}$) maps causal curves to causal curves. Let $t\to x(t)$ 
$(a\leq t\leq b)$ be a causal curve in $K$. Since by definition the stages $\Psi_s^Z$ of the local flow of $Z$ are isometries we surely have that 
$0 \geq \langle x'(t_0),x'(t_0)\rangle_{x(t_0)} = \lim_{t\to 0} \frac{1}{t^2}\langle x(t_0+t)-x(t_0), x(t_0+t)-x(t_0)\rangle_{x(t_0)} =  \lim_{t\to 0} 
\frac{1}{t^2}\langle d(\Psi_s)_{x(t_0)}(x(t_0+t))-d(\Psi_s)_{x(t_0)}(x(t_0)), d(\Psi_s)_{x(t_0)}(x(t_0+t))-d(\Psi_s)_{x(t_0)}(x(t_0))\rangle_{\Psi_s(x(t_0))} = 
\langle d(\Psi_s)_{x(t_0)}(x'(t_0)),d(\Psi_s)_{x(t_0)}(x'(t_0))\rangle_{\Psi_s(x(t_0))}$. So stages of local flows of Killing fields preserve causal 
curves. 
\end{proof}

\begin{proposition}\label{locflow} Let $\Psi_s^Z$ be the stages of a local flow of a Killing vector field $Z$. For any $\mathcal{O}\in \mathcal{K}(M,g)$ there 
exists some $0<t_0$ so that $\Psi_s^Z$ is defined on $\mathcal{O}_\cup=\cup_{s\in[-t_0,t_0]}\Psi_s^Z(\mathcal{O})$ for all $s\in [-t_0,t_0]$ with $\Psi_{s+r}^Z=\Psi_s^Z\Psi_r^Z$ 
on $\mathcal{O}$ whenever $-t_0<s, r, s+r<t_0$ and with each $\Psi_s$ $(-t_0<s<t_0)$ mapping $\mathcal{O}$ onto another element of $\mathcal{K}(M,g)$.
\end{proposition}

\begin{proof} The most difficult property to prove is the preservation of causal convexity for $t$ in some interval. The rest of the claims are fairly direct 
consequences of the description in for example \S 9.3 of \cite{GQ1}. We proceed with proving the claim regarding causal convexity. Let $K$ be a compact set and 
select $t>0$ so that $\Psi_s^Z$ is defined on $K$ for all $s\in [-t,t]$ with $\Psi_{s+r}^Z=\Psi_s^Z\Psi_r^Z$ on $K$ whenever $-t<s, r, s+r<t$. By Tychonoff's 
theorem $\Pi_{s\in[-t,t]}\Psi_s^Z(K)$ is compact. We recall that $\Pi_{s\in[-t,t]}\Psi_s^Z(K)$ is equipped with the coarsest topology for which all the 
projections $\pi_t\colon \Pi_{s\in[-t,t]}\Psi_s^Z(K)\to \Psi_t^Z(K)$ are continuous \cite[Example 3, Ch 1, \S2.3]{Bour}. That in turn ensures that the function 
$\iota\colon \Pi_{s\in[-t,t]}\Psi_s^Z(K)\to \cup_{s\in[-t,t]}\Psi_s^Z(K)$ which for every $x\in \Psi_t^Z(K)$ is given by $\iota(t,x)=x\in \Psi_t^Z(K)$, is a 
continuous surjection. But then the image $\cup_{s\in[-t,t]}\Psi_s^Z(K)$ is itself compact. We may take $K$ to be the closure of some $\mathcal{O}\in \mathcal{K}(M,g)$. 
By Proposition \ref{directed} we may then select another element $\mathcal{O}_{\cup}\in\mathcal{K}(M,g)$ containing $\cup_{s\in[-t,t]}\Psi_s^Z(K)$. Now pick 
$0<t_0<t$ so that $\Psi_s^Z$ is defined on $\mathcal{O}_{\cup}$ for all $s\in [-t_0,t_0]$ with $\Psi_{s+r}^Z=\Psi_s^Z\Psi_r^Z$ on $\mathcal{O}_{\cup}$ whenever 
$-t_0<s, r, s+r<t_0$. Given $p,q\in \mathcal{O}$ the causal convexity of $\mathcal{O}_{\cup}$ ensures that any causal curve $\gamma$ joining $\Psi_{t_0}^Z(p)$ and 
$\Psi_{t_0}^Z(q)$ is contained in $\mathcal{O}_{\cup}$. Since $\Psi_{-t_0}^Z$ is defined on $\mathcal{O}_{\cup}$ it will map the causal curve $\gamma$ onto a 
causal curve $\Psi_{-t_0}^Z(\gamma)$ joining $p$ and $q$. The causal convexity of $\mathcal{O}$ ensures that this curve is contained in $\mathcal{O}$ joining $p$ 
and $q$. But then $\gamma=\Psi_{t_0}^Z(\Psi_{-t_0}^Z(\gamma))$ is a causal curve inside $\Psi_{t_0}^Z(\mathcal{O})$ joining $\Psi_{t_0}^Z(p)$ and 
$\Psi_{t_0}^Z(q)$. So $\Psi_{t_0}^Z(\mathcal{O})$ is causally convex.
\end{proof}

\section{Killing flows on local algebras: an illustrative example based on the Klein-Gordon equation}\label{KillingFlows}

\begin{remark}\label{BGP-localg}
We will use the ideas in Lemma 4.4.8 and Corollary 4.4.12 of \cite{BGP} to construct a local algebra from the collection $\mathcal{K}(M,g)$. Each local algebra 
$\cA(\mathcal{O})$ will be the Weyl algebra corresponding to the solutions of the localised version of the Klein-Gordon equation as described in Proposition 3.5.1 
of \cite{BGP}. We remind the reader that a Weyl algebra is a unital $C^*$-algebra generated by a set of unitaries $\{W(f)\colon f\in \cH\}$ governed by the 
relations $W(f)W(g)=e^{-i\sigma(f,g)/2}W(f+g)$ and $W(f)^*=W(-f)$, where $\cH$ is a real symplectic space equipped with a non-singular symplectic form $\sigma$. 

To see the claim that the local algebras mentioned above are Weyl algebras of solutions of the Klein-Gordon equation, note that Theorem 3.2.15 of \cite{BD} confirms that in the 
globally hyperbolic case the quotient space described prior to \cite[Lemma 4.3.8]{BGP} may be thought of as a space of solutions of the Klein-Gordon equation. Let $\mathcal{O}$, 
$t_0$ and $\Psi_s^Z$ be as in the previous Proposition. The remaining challenge is to show that for any stage $\Psi_s^Z$ $(0<s\leq t_0)$, we will for the local algebras 
constructed according to the prescription in Lemma 4.4.8 and Corollary 4.4.12 of \cite{BGP} have that the map $\Psi_s^Z\colon \mathcal{O}\to M$ canonically lifts to a 
$*$-isomorphism $\pi_s^Z$ from $\cA(\mathcal{O})$ to $\cA(\Psi_s^Z(\mathcal{O}))$ with $\pi_s^Z\circ\pi_r^Z=\pi_{r+s}^Z$ whenever the same is true for the $\Psi_s^Z$'s. 
This seems to only be possible for the class of sets $\mathcal{K}(M,g)$ described in \cite{BFV}.
\end{remark}

\medskip

\textbf{Note:} For the rest of this section $\cA(M)$ will be the local algebra constructed from the collection $\mathcal{K}(M,g)$ according to the ideas in 
Lemma 4.4.8 and Corollary 4.4.12 of \cite{BGP} 

\medskip

\begin{lemma}\label{Green} Let $(M,g)$ be globally hyperbolic and let $\Psi$ be an isometric diffeomorphism which maps $\mathcal{O}\in\mathcal{K}(M,g)$ onto 
another element $\Psi(\mathcal{O})$ of $\mathcal{K}(M,g)$. Let $P_{\mathcal{O}}$ and $P_{\Psi(\mathcal{O})}$ be the differential operators corresponding to the 
Klein-Gordon equation on $\mathcal{O}$ and $\Psi(\mathcal{O})$ respectively. Let $\widetilde{G}_{\mathcal{O}}^\pm$ (resp. $\widetilde{G}_{\Psi(\mathcal{O})}^\pm$) 
be the corresponding Green's operators as described in \cite[Prop 3.5.1]{BGP}. 
\begin{itemize} 
\item If $\Psi^{-1}$ preserves time-orientation, then $f\to\widetilde{G}^+_{\Psi(\mathcal{O})}(f\circ\Psi^{-1})\circ\Psi$ (where 
$f\in \mathscr{D}(\mathcal{O}, E)$) agrees with $\widetilde{G}^+_{\mathcal{O}}$.
\item If $\Psi^{-1}$ reverses time-orientation, then $f\to\widetilde{G}^+_{\Psi(\mathcal{O})}(f\circ\Psi^{-1})\circ\Psi$ (where 
$f\in \mathscr{D}(\mathcal{O}, E)$) agrees with $\widetilde{G}^-_{\mathcal{O}}$.
\end{itemize}
\end{lemma}

\begin{proof} We observe that $\mathcal{O}, \Psi(\mathcal{O}) \in \mathcal{K}(M,g)$ are in effect globally hyperbolic Lorentzian manifolds in their own right. 
We therefore have access to the results of Dimock \cite{Dim} who showed on page 226 of his paper that $\Psi^*P_{\mathcal{O}}=P_{\Psi(\mathcal{O})}\Psi^*$ where 
$\Psi^*$ is the push-forward. Given $f\in\mathscr{D}(\mathcal{O},E)$ this will by \cite[Def 3.4.1]{BGP} mean that 
$$\widetilde{P}_{\mathcal{O}}(\widetilde{G}^\pm_{\Psi(\mathcal{O})}(f\circ\Psi^{-1})\circ\Psi)= P_{\Psi(\mathcal{O})}\widetilde{G}^\pm_{\Psi(\mathcal{O})}(f\circ\Psi^{-1})\circ\Psi$$ 
$$= PG(f\circ\Psi^{-1}_{ext}){\upharpoonright}{\Psi(\mathcal{O})}\circ\Psi= (f\circ\Psi^{-1}_{ext}){\upharpoonright}{\Psi(\mathcal{O})}\circ\Psi=f.$$Similarly 
$$\widetilde{G}^\pm_{\Psi(\mathcal{O})}(\widetilde{P}_{\mathcal{O}}(f)\circ\Psi^{-1})\circ\Psi = \widetilde{G}^\pm_{\Psi(\mathcal{O})}(\widetilde{P}_{\Psi(\mathcal{O})}(f\circ\Psi^{-1}))\circ\Psi$$ 
$$=G(P(f\circ\Psi^{-1})_{ext})\circ\Psi = (f\circ\Psi^{-1})_{ext}\circ\Psi=f.$$If $\Psi^{-1}$ preserves time-orientation, then since it also preserves causal 
curves, we may therefore conclude from \cite[Prop 3.5.1]{BGP} that $$\mathrm{supp}(\widetilde{G}^\pm_{\Psi(\mathcal{O})}(f\circ\Psi^{-1})\circ\Psi) \subseteq 
\Psi^{-1}(\mathrm{supp}(\widetilde{G}^\pm_{\Psi(\mathcal{O})}(f\circ\Psi^{-1})))\subset \Psi^{-1}(J^\pm_{\Psi(\mathcal{O})}(f\circ\Psi^{-1}))= J^\pm_{\mathcal{O}}(f).$$Thus 
$f\to \widetilde{G}^\pm_{\Psi(\mathcal{O})}(f\circ\Psi^{-1})\circ\Psi$ is an advanced Green's operator for the action of 
$P_\mathcal{O}$ on $\mathcal{O}$. But for globally hyperbolic Lorentzian manifolds such operators are unique, which then proves the claim in this case. The proof 
of the case where $\Psi^{-1}$ reverses time-orientation is similar.
\end{proof}

\begin{theorem}\label{localg-ex}
Let $(M,g)$ be a globally hyperbolic and let $\Psi$ be an isometric diffeomorphism which maps $\mathcal{O}\in\mathcal{K}(M,g)$ onto another element 
$\Psi(\mathcal{O})$ of $\mathcal{K}(M,g)$. 
\begin{enumerate}
\item For the local algebra described in part (1) of \cite[Theorem 3.4]{LM} the prescription $f\to f\circ \Psi^{-1}$ where $f\in C^\infty_0(\mathcal{O})$ lifts to 
a $*$-isomorphism $\pi_\Psi$ from $\cA(\mathcal{O})$ to $\cA(\Psi(\mathcal{O}))$ which sends $W(f)$ where $f\in C^\infty_0(\mathcal{O})$ to $W(f\circ \Psi^{-1})$.
\item If $\Psi^{-1}$ preserves time-orientation then for the local algebra described in \cite[Lemmata 4.4.8 \& 4.4.10]{BGP} we also have that the map 
$\Psi\colon \mathcal{O}\to \Psi(\mathcal{O})$ canonically lifts to a $*$-isomorphism $\pi_\Psi$ from $\cA(\mathcal{O})$ to $\cA(\Psi(\mathcal{O}))$ which 
implements $\Psi$ in the sense of sending $W(\widetilde{G}_{\mathcal{O}}(f))$ where $f\in \mathscr{D}(\mathcal{O}, E)$ to 
$W(\widetilde{G}_{\Psi(\mathcal{O})}(f\circ\Psi^{-1}))$.
\item If $\Psi^{-1}$ reverses time-orientation then for the local algebra described in \cite[Lemmata 4.4.8 \& 4.4.10]{BGP} we have that the map $\Psi\colon \mathcal{O}\to \Psi(\mathcal{O})$ 
canonically lifts to a $*$-anti-isomorphism $\pi_\Psi$ from $\cA(\mathcal{O})$ to $\cA(\Psi(\mathcal{O}))$ which implements $\Psi$ in the above sense, but this time sending 
$W(\widetilde{G}_{\mathcal{O}}(f))$ to 
$W(-\widetilde{G}_{\Psi(\mathcal{O})}(f\circ\Psi^{-1}))$.
\end{enumerate}
\end{theorem}

\begin{proof}
\textbf{Claim 1:} Firstly note that \cite[Prop 5.19]{GQ2} suggests that if $d\mu$ is the measure coming from the volume form, then on $\mathcal{O}$ the stage 
$\Psi_s$ is measure preserving. So for any $f, g\in C_0^\infty(M)$ with support in $\mathcal{O}$ we have that $\int_{\mathcal{O}}f\overline{g}\,d\mu = 
\int_{\Psi(\mathcal{O})}(f\circ\Psi^{-1})(\overline{g\circ\Psi^{-1}})\,d\mu$ where we have used the fact that $f\circ\Psi^{-1}$ and $g\circ\Psi^{-1}$ are elements 
of $C_0^\infty(M)$ supported on $\Psi(\mathcal{O})$. The claim now follows from \cite[Theorem 4.2.9]{BGP}.

\medskip 

\textbf{Claims 2 and 3:} Suppose that $\Psi$ either reverses or preserves time-orientation, and let $\widetilde{P}_{\mathcal{O}}$ be the Klein-Gordon operator and 
$\widetilde{G}_{\mathcal{O}}$ the Green's operator associated with $\mathcal{O}$ as described in \cite[Prop 3.5.1]{BGP}. It is now clear from Lemma \ref{Green} 
that the prescription $f\to f\circ\Psi^{-1}$ will map $\widetilde{G}^+_{\mathcal{O}}(f)-\widetilde{G}^-_{\mathcal{O}}(f)$ onto either 
$\widetilde{G}^+_{\Psi(\mathcal{O})}(f\circ\Psi^{-1})\circ\Psi-\widetilde{G}^-_{\Psi(\mathcal{O})}(f\circ\Psi^{-1})\circ\Psi$ or 
$\widetilde{G}^-_{\Psi(\mathcal{O})}(f\circ\Psi^{-1})\circ\Psi-\widetilde{G}^+_{\Psi(\mathcal{O})}(f\circ\Psi^{-1})\circ\Psi$, depending on whether $\Psi$ either 
preserves or reverses time-orientation. 

\medskip

\emph{Case 1 - preservation of time-orientation:} If $\Psi$ preserves time orientation then the facts about preservation of 
measure noted in the first part of the proof, ensure that \begin{eqnarray*}
\int_M\langle g, \widetilde{G}_{\mathcal{O}}(f)\rangle\,dV 
&=& \int_M\langle g, \widetilde{G}_{\Psi(\mathcal{O})}(f\circ\Psi^{-1})\circ\Psi\rangle\,dV\\
&=& \int_M\langle g\circ\Psi^{-1}, \widetilde{G}_{\Psi(\mathcal{O})}(f\circ\Psi^{-1})\rangle\,dV.
\end{eqnarray*}
Since $\Psi$ preserves time-orientation, the prescription $\psi\to \psi\circ\Psi^{-1}$ yields a symplectic map mapping 
$\{\widetilde{G}_{\mathcal{O}}(f)\colon f\in \mathscr{D}(\mathcal{O},E)\}$ onto $\{\widetilde{G}_{\Psi(\mathcal{O})}(f)\colon 
f\in\mathscr{D}(\Psi(\mathcal{O}),E)\}$. To in this case obtain the conclusion, all that is required is to apply Theorem 4.2.9 
of \cite{BGP} to the above fact.

\medskip

\emph{Case 2 - reversal of time-orientation:} If on the other hand $\Psi$ reverses time-orientation, a similar computation shows that 
$\int_M\langle g, \widetilde{G}_{\mathcal{O}}(f)\rangle\,dV = -\int_M\langle g\circ\Psi^{-1}, \widetilde{G}_{\Psi(\mathcal{O})}(f\circ\Psi^{-1})\rangle\,dV$. To 
obtain the conclusion we will in this case need to work a bit harder. So here the prescription $f\to f\circ\Psi^{-1}$ yields a *-isomorphism from 
$\cA(\mathcal{O})$ to the algebra $\widetilde{\cA}(\Psi(\mathcal{O}))$ constructed from the degenerate symplectic form 
$(\varphi,\varrho)\to -\int_M\langle \varphi, \widetilde{G}_{\Psi(\mathcal{O})}(\varrho)\rangle\,dV$ on $\mathscr{D}(\Psi(\mathcal{O}),E)$, which maps 
$W(\widetilde{G}_{\mathcal{O}}(f))$ to $\widetilde{W}(\widetilde{G}_{\Psi(\mathcal{O})}(f\circ\Psi^{-1}))$ (where $\widetilde{W}(\varphi)$ are the Weyl operators 
generated using the above symplectic form. A careful consideration of the third displayed equation in \cite[Example 4.2.2]{BGP} shows that for any 
$\varphi \in\mathscr{D}(\Psi(\mathcal{O}),E)$, the Weyl operator $\widetilde{W}(\varphi)$ described above corresponds to $W(-\varphi) = W(\varphi)^*$. Thus the 
algebras $\cA(\Psi(\mathcal{O}))$ and  $\widetilde{\cA}(\Psi(\mathcal{O}))$ agree as they are generated by the same operators. The *-isomorphism obtained above 
then sends $W(\widetilde{G}_{\mathcal{O}}(f))$ to $W(-\widetilde{G}_{\Psi(\mathcal{O})}(f\circ\Psi^{-1})) = 
W(\widetilde{G}_{\Psi(\mathcal{O})}(f\circ\Psi^{-1}))^*$. By now following this *-isomorphism with involution, we obtain the promised *-anti-isomorphism. (Some 
technicalities have been suppressed for the sake of clarity. For full details see pages 129-133 of \cite{BGP}.)
\end{proof}

\begin{remark} 
If $(M,g)$ is connected then all local diffeomorphisms either preserve or reverse time-orientation. See the discussion preceding \cite[Cor 7.9]{ONeil}.
\end{remark}

\begin{lemma}\label{time-orient} Let $\Psi_s^Z$ be the stages of a local flow of a Killing vector field $Z$ and let $\mathcal{O}\in \mathcal{K}(M,g)$ be given. 
Let $0<t_0$ be as in Propositions 4.10. For any $s\in [-t_0, t_0]$ the stage $\Psi_s^Z$ will then preserve time-orientation on $\mathcal{O}$.
\end{lemma}

\begin{proof} We remind the reader that $\mathcal{O}$ is globally hyperbolic in its own right and that globally hyperbolic Lorentzian manifolds are connected by definition (see \cite{BGP}). 
It is clear from the preceding remark that $\Psi_{s/2}^Z$ will for any $s\in [-t_0, t_0]$ either preserve or reverse time-orientation. But in either case $\Psi_s^Z= \Psi_{s/2}^Z\circ \Psi_{s/2}^Z$ 
will then preserve time-orientation. The claim therefore follows.
\end{proof}

Theorem \ref{localg-ex} and Lemma \ref{time-orient} now yields the following result:

\begin{corollary}\label{loc-flow} Let $(M,g)$ be a globally hyperbolic manifold and let $Z$ be a Killing vector field. Let $\mathcal{O}\in \mathcal{K}(M,g)$. 
Select a positive real $t_0>0$ so that $\Psi_s^Z$ is defined on $\mathcal{O}_{\cup}$ for all $s\in [-t_0,t_0]$ with $\Psi_{s+r}^Z=\Psi_s^Z\Psi_r^Z$ on 
$\mathcal{O}_{\cup}$ whenever $-t_0<s, r, s+r<t_0$ and with each $\Psi_s^Z$ $(-t_0<s<t_0)$ mapping $\mathcal{O}$ onto a relatively compact causally compatible 
globally hyperbolic open subset of $M$. 
\begin{enumerate}
\item For any stage $\Psi_s^Z$ $(0<s\leq t_0)$ there exists a $*$ isomorphism $\pi_s^Z$ from 
$\cA(\mathcal{O})$ to $\cA(\Psi_s^Z(\mathcal{O}))$ which implements the action of the local flow at the algebraic level. 
\item On passing to the Fock representation of these Weyl algebras one typically encounters in specific examples (see for example \cite{Dim}, \cite[\S 4.7]{BGP} 
and \cite[5.2.24]{KM}), we note that for any $a\in \cA(\mathcal{O})$ we have that $\pi_s(Z)(a)$ strongly converges to $a$ as $s\to 0$. 
\end{enumerate}
\end{corollary}

\begin{proof} Only the final claim needs to be proven. We know from \cite[Proposition 4.6.10]{BGP} that  
$\pi_s^Z(W(\widetilde{G}_{\mathcal{O}}(f)))=W(\widetilde{G}_{\Psi_s^Z(\mathcal{O})}(f\circ(\Psi_s^Z)^{-1}))= W(\widetilde{G}_{\mathcal{O}}(f)(\Psi_s^Z)^{-1})$ is 
strongly convergent to $W(\widetilde{G}_{\mathcal{O}}(f))$ on Fock space as $s\to 0$. The same is therefore true for the span of the Weyl operators, which is norm 
dense in the $\mathrm{CCR}$ algebra $\cA(\mathcal{O})$. So by suitably approximating we can show that for any $a\in \cA(\mathcal{O})$, 
$\pi_s^Z(a)\to a$ strongly as $s\to 0$.
\end{proof}

\section{Quantum Killing Lie derivatives}\label{LieDer} 

We shall here pass to the von Neumann algebra setting assuming all our von Neumann local algebras $\cM(\mathcal{O})$ to be the $\sigma$-weak closure of the 
matching $C^*$-algebras $\cA(\mathcal{O})$ in some appropriate representation. (We shall elaborate more on what we mean by ``appropriate representations'' in the 
next part of the paper.) In the notation of Corollary \ref{loc-flow} we shall also assume that for any $a\in \cM(\mathcal{O})$ we have that $\pi_s^Z(a)$ strongly 
converges to $a$ as $s\to 0$. To see that this holds in the Fock space case alluded to in Corollary \ref{loc-flow}, notice that by \cite[Theorem 6]{Der} the Bogoliubov 
transformations which are used to construct the $\pi^Z_t$s are implemented by a unitary. Thus the $\pi_t^Z$s admit a canonical extension to the double commutants of 
the $\cA(\mathcal{O})$s. In addition since each of the $\pi_t^Z$s described in Corollary \ref{loc-flow} is contractive, the convergence of $(\pi_t^Z(a))$ noted there 
is actually in the $\sigma$-strong topology. By convexity each $\cA(\mathcal{O})$ is $\sigma$-strongly dense in $\cM(\mathcal{O})$, and hence in this case the 
convergence we are assuming is just a continuous extension of the convergence noted in part (2) of Corollary \ref{loc-flow}. We pass to considering the consequences of 
strong continuity at 0. Formally Quantum Killing Lie derivatives are the infinitesimal generators of the quantum local flows of Killing vector fields. These quantum Lie 
derivatives are defined as follows: 

\begin{definition} Let $(M,g)$ be globally hyperbolic. Given a Killing vector field $Z$ on $(M,g)$, we define $\delta_Z$ to be the operator 
$\cup_{\mathcal{O}\in \mathcal{K}(M,g)}\cM(\mathcal{O})\to \cM(M)$ with domain 
$\mathrm{dom}(\delta_Z)=\{f\in \cM(\mathcal{O})\colon \mathcal{O}\in \mathcal{K}(M,g), \lim_{t\to 0}\frac{1}{t}(\pi_t^Z(f)-f)\mbox{ exists as an element of }\cM(M)\}$ and with values given by 
$\delta_Z(f)= \lim_{t\to 0}\frac{1}{t}(\pi_t^Z(f)-f)$.
\end{definition}

There is a very well-developed theory of semigroups of operators which we may use as a guide in our development. However we are here not dealing with semigroups and there is no obvious 
way to reduce the questions we wish to answer to the semigroup setting. The tools we shall need in our analysis therefore need to be developed from scratch. We proceed with developing 
the technical background we require to get positive results.

\subsection{Generators of one-parameter local groups of contractions}

We are ultimately interested in describing the generators of the quantum Killing flows described in the previous subsection. To be able to achieve that 
objective the appropriate machinery needs to be developed. This subsection is devoted to laying that groundwork.

\begin{definition}\label{lcg-def} We call a set $\{T_t\colon t\in \mathbb{R}\}$ of partially defined linear contractions on a von Neumann algebra $\cM$ a strongly continuous 
one-parameter local group of contractions if there exists an associated collection $\{D_\alpha\colon \alpha\in (0,\infty)\}$ of subspaces such that 
\begin{enumerate}
\item $D_\alpha\subseteq D_\beta$ if $0<\beta\leq \alpha$ with $\cup_{\alpha>0}D_\alpha$ strongly dense in $\cM$;
\item $D_\alpha \subseteq\mathrm{dom}(T_t)$ if $|t|\leq \alpha$ with $T_0=\mathrm{Id}$;
\item for every $x\in D_\alpha$ we have $T_sT_t(x)=T_{s+t}(x)$ whenever $s, t, s+t$ all belong to $[-\alpha, \alpha]$; 
\item for every $x\in \cup_{\alpha>0}D_\alpha$ we have that $T_t(x)\to x$ strongly as $t\to 0$.
\end{enumerate}
Note that for large $\alpha>0$, the subspace $D_\alpha$ may just be $\{0\}$.  
\end{definition}

\begin{remark}\label{aaa}
Since for each $x\in D_\alpha$ the set $\{T_t(x)\colon t\in(-\alpha,\alpha)\}$ is bounded, the convergence in part (4) above is actually $\sigma$-strong. 
If the $T_t$s preserve adjoints the convergence is even $\sigma$-strong*.
\end{remark}

\begin{lemma}\label{lcg-lem1} Let $\{T_t\colon t\in \mathbb{R}\}$ be a strongly continuous 1-parameter local group of contractions on $\cM$. For every $\alpha>0$ 
and every $x\in D_\alpha$ the function $t\to T_t(x)$ is then a point to $\sigma$-strong continuous function from $[-\alpha,\alpha]$ to $\cM$.
\end{lemma}

\begin{proof} A straightforward modification of the proof of \cite[Corollary 2.3]{Pazy} suffices to prove the result. Continuity at $t=0$ follows from the 
definition. For $t\in (-\alpha,\alpha)$ we may simply replace $x$ with $T_t(x)$ and apply the claims in the definition to the limit $\lim_{s\to 0}T_s(T_t(x))$.  
The endpoints need to be handled with more care as one needs to check for one-sided continuity. If for example $t=-\alpha$ then for $s>0$ small enough we will for 
any $\xi \in H$ have that $\|T_{-\alpha}(x)\xi-T_{-\alpha+s}(x)\xi\|\leq \|T_{-\alpha}\|.\|x(\xi)-T_s(x)(\xi)\|\to 0$ as $s\searrow 0$. 
\end{proof}

\begin{lemma} Let $\{T_t\colon t\in \mathbb{R}\}$ be a strongly continuous 1-parameter local group of contractions on $\cM$ and let $A$ be its infinitesimal 
generator. For every $x\in \cup_{\alpha>0}D_\alpha$ the following will hold for every $t$ with $|t|$ small enough:
\begin{enumerate}
\item $\lim_{h\searrow 0}\frac{1}{h}\int_t^{t+h}T_s(x)\,ds= \lim_{h\nearrow 0}\int_{t-h}^tT_s(x)\,ds =T_t(x)$.
\item $\int_0^{t}T_s(x)\,ds\in \mathrm{dom}(A)$ with $A\int_0^{t}T_s(x)\,ds=T_t(x)-x$.
\end{enumerate}
\end{lemma}

\begin{proof} Given $x\in \cup_{\alpha>0}D_\alpha$ we clearly have that $x\in D_\beta$ for some $\beta>0$. In view of the continuity noted in Lemma 
\ref{lcg-lem1}, part (1) follows from a straightforward generalisation a well known fact regarding the Riemann integral . It remains to prove (2). Given 
$t\in (0,\beta)$, it then follows from Definition \ref{lcg-def} and Lemma \ref{lcg-lem1} that we will for any $h\in (0,\beta-t)$ which is small enough so 
that $|h|<|t|$ have that
$$\frac{1}{h}(T_h-\I)\int_0^tT_s(x)\,ds = \frac{1}{h}\int_0^t(T_{s+h}-T_s)(x)\,ds = \frac{1}{h}\left[\int_t^{t+h}T_s(x)\,ds-\int_0^hT_s(x)\,ds\right].$$As 
$h\searrow 0$, the right hand side will by part (1) converge to $T_t(x)-x$. The claim now holds by definition.
\end{proof}

\begin{corollary} Let $\{T_t\colon t\in \mathbb{R}\}$ be a strongly continuous 1-parameter local group of contractions on $\cM$. Then the infinitesimal generator 
is $\sigma$-strong densely defined. If $\{T_t\}$ preserves adjoints it is even $\sigma$-strong* densely defined.
\end{corollary}

\begin{proof}  This follows from a combination of the previous lemma and Remark \ref{aaa}.
\end{proof}

\begin{definition} For a $\sigma$-strong* densely defined operator $T\colon\cM\supset \mathrm{dom}(T)\to\cM$ we define the $\sigma$-strong* separating space of 
$T$ to be $S(T)=\{b\in \cM\colon \mbox{ there exists a net }(a_\gamma)\subset \mathrm{dom}(T)\mbox{ such that }a_\gamma\to 0 
\mbox{ and }T(a_\gamma)\to b \quad\sigma-\mathrm{strong}*\}$. 
\end{definition}

For $T$ as above it is clear that $T$ will be $\sigma$-strong* closable (equivalently $\sigma$-weakly closable by convexity) iff $S(T)=\{0\}$. We use this fact to 
posit a closability criterion of *-derivations. The proof is based on an idea of Niknam \cite{Nik}. Closability criteria of this nature seem to have first been 
recorded by Bratteli and Robinson (see \cite[Theorem 4]{BR3}) using a quite different proof. 

\begin{lemma} Let $\delta$ be a $\sigma$-strong* densely defined *-derivation on $\cM$.
\begin{enumerate}
\item Then $\overline{S(\delta)}^{w*}$ is a two sided *-ideal.
\item If there exists a normal state $\omega$ which is faithful on the centre of $\cM$ and for which we have that $\omega\circ\delta=0$, then $\delta$ is $\sigma$-weakly closable.
\end{enumerate}
\end{lemma}

\begin{proof} The first claim is an easy consequence of the product rule for derivations. Specifically given any $b\in S(\delta)$, $b_0\in \mathrm{dom}(\delta)$ 
and $(a_\gamma)\subset \mathrm{dom}(\delta)$ with $a_\gamma\to 0$ and $\delta(a_\gamma)\to b$, we have that $a_\gamma b_0, b_0a_\gamma \in \mathrm{dom}(\delta)$. 
Moreover both $(b_0a_\delta)$ and $(a_\gamma b_0)$ converge $\sigma$-strong* to 0 with $\delta(b_0a_\gamma)=\delta(b_0)a_\gamma +b_0\delta(a_\gamma)\to b_0b$ and 
similarly $\delta(a_\gamma b_0)\to bb_0$. So $S(\delta)$ is an $\mathrm{dom}(\delta)$-ideal. Now let $c\in \cM$ be given. By the density of $\mathrm{dom}(\delta)$ 
we may select a net $(b_\gamma)\subset \mathrm{dom}(\delta)$ converging $\sigma$-strong* to $c$. From what we have just shown, $(bb_\gamma), (b_\gamma b)\subset  S(\delta)$. 
Since by convexity the $\sigma$-weak and $\sigma$-strong* closures of $S(\delta)$ agree, it follows that $bc, cb\in \overline{S(\delta)}^{w*}$.
 
For the second claim note that since $\overline{S(\delta)}^{w*}$ is a weak* closed ideal, there must exists a central projection $e$ such that 
$e\cM=\overline{S(\delta)}^{w*}$. Note that by definition $S(\delta)\subseteq\overline{\mathrm{ran}(\delta)}^{w*}$. (Here we used the fact that by convexity the $\sigma$-weak 
and $\sigma$-strong* closures of $\mathrm{ran}(\delta)$ agree.) So if there is indeed a normal state such that $\omega\circ\delta=0$, then that state annihilates $\overline{S(\delta)}^{w*}$. 
Faithfulness of $\delta$ on the centre would then ensure that $e=0$, which then forces $S(\delta)=\{0\}$.
\end{proof}

\subsection{Existence of Quantum Killing Lie derivatives}

\begin{remark}\label{QLder} Let $Z$ be a Killing vector field on $(M,g)$. Given $\mathcal{O}\in \mathcal{K}(M,g)$ we know from Corollary \ref{loc-flow} that there 
exists $t_0 >0$ such that for all $t\in [-t_0, t_0]$ the *-homomorphisms $\pi_t^Z$ are defined on $\cM(\mathcal{O})$ and implement the action of the local flow 
$\Psi_t^Z$ as described in Proposition \ref{locflow}. For a fixed $s$ let $\mathscr{T}_s$ be the set of all $\mathcal{O}\in\mathcal{K}(M,g)$ for which the 
*-homomorphisms $\pi_t^Z$ are defined on $\cM(\mathcal{O})$ and implement the action of the local flow $\Psi_t^Z$ for all $t\in [-s, s]$. Now let $D_S$ be 
$\mathrm{span}\{\cM(\mathcal{O})\colon \mathcal{O}\in \mathscr{T}_s\}$. The continuity of the $\pi_t^Z$'s ensure that the $D_s$'s satisfy the criteria of 
Definition \ref{lcg-def}.
\end{remark}

\begin{theorem}\label{LieDer-exist}  Let $(M,g)$ be a globally hyperbolic Lorentzian manifold and let $\cM(M)$ be as before. If for each Killing vector field $Z$ 
and each $\mathcal{O}\in \mathcal{K}(M,g)$ there exists a normal state $\omega_{\mathcal{O}}$ on $\cM(\mathcal{O})$ which is faithful on the centre of $\cM(\mathcal{O})$
and for which we have that $\omega_{\mathcal{O}}\circ\delta_Z=0$, the operator $\delta_Z$ is then a $\sigma$-strong* densely defined closable unital *-derivation.
\end{theorem}

\begin{proof}
Remark \ref{QLder} ensures that we have access to the theory of the previous section. Let $\cM(\mathcal{O})$ be given. The fact that $\delta_Z$ is a derivation 
follows by applying the product rule to the fact that we will for small enough $t$ have that $\pi_t^Z(ab)=\pi_t^Z(a)\pi_t^Z(b)$. Similarly the fact that 
$\pi_t^Z(a^*)=\pi_t^Z(a)^*$ ensures that $\delta_Z$ is a $*$-derivation. In view of the strong continuity noted above, the fact that $\delta_Z$ is 
$\sigma$-strong* densely defined follows from the results of the previous subsection. Then the results in the previous subsection also ensure that $\delta_Z$ is 
$\sigma$-strong* closable on $\cM(\mathcal{O})$. 
\end{proof}

\part{Von Neumann algebraic representations of local algebras}\label{Pt.2}

\section{The selection of appropriate von Neumann algebraic representations}\label{S7}

In this Section we will describe in detail von Neumann algebras related to subsets $\cO \in \cK(M,g)$ as well as mappings between them with some emphasis on dynamical maps originating 
from the Klein-Gordon equation. We start with selecting the appropriate algebras that describe the fields. The key observation which makes this selection possible is that the Weyl algebra 
constructed from the proper symplectic space encompasses, via affiliation, the fundamental (unbounded) observables for Quantum Field Theory in the curved spacetime. 
We shall here look specifically at regular (and in particular quasi-free) representations of these algebras. We recall that elements of the Weyl algebra, Weyl operators, are of the form
\begin{equation}
W(f) = e^{i \phi(f)},
\end{equation}
where $\phi(f)$ is the smeared field operator satisfying the commutation relation
$$(\phi(f)\phi(g) - \phi(g)\phi(f)) \psi 
= \sigma(f,g) \psi \mbox{ for all }f,g \in \cH, \psi \in \rm{dom}(\phi(f)),$$and $\cH$ is a symplectic space. As was already mentioned in the introduction, 
the description of a system based on von Neumann algebras includes 
the algebraic structure of observables, while the GNS construction 
induced by the given state $\omega$ describes the above structure in the representation reflecting the current state of the system.
However, when describing a system with curved spacetime, an example of dynamic relations 
 given by the (semi-classical) Einstein equation
should be taken into account, cf. \cite{Wald2}. It is assumed that the back-reaction of the quantum field on the spacetime
is described by the equation
\begin{equation}
G_{ab} = 8\pi \langle T_{ab}\rangle_{\omega},
\end{equation}
where $T_{ab}$ is the stress-energy tensor and $G_{ab}$ the Einstein tensor. It is important to note that $T_{ab}$ can be expressed as 
a quadratic form of the (unbounded) field operators. Therefore, even when regarding the semiclassical Einstein equation as an effective approximation rather than a 
complete theory of gravity coupled to quantum matter, this equation clearly points to the necessity of including unbounded observables in the description.

Based on recent results, we can say that there are 
two possibilities to enlarge the set of observables for the 
description of the system in question. It is worth noting that if the set of observables is larger, then 
the set of states (continuous functionals on that set) consists of more regular states. An illustration of this claim was given in \cite{ML2}.

The first option is to limit ourselves to a special class of states, specifically to the subclass of quasi-free states 
that satisfy Hadamard's condition.
The second is to use Orlicz space techniques as was described in our recent paper \cite{LM}.
So in this Section we will limit ourselves to quasi-free states with some emphasis on Hadamard states.

Within the framework of algebraic quantization of field theory, we are not interested in using microlocal analysis to describe states that may be 
analogues of the vacuum in Minkowski space. As far as states are concerned, we are rather interested in  those states whose GNS representation well 
illustrates the presented theory. Among such states the quasi-free states are of course particularly significant in view of the fact that their 
GNS representation yield Araki-Woods factors. Within this class the Hadamard states are those which are widely used in QFT in particular due to the 
fact that the expectation value of the quantum stress energy tensor corresponding to a Hadamard state makes sense. As we shall shortly see, there is no 
need to search far and wide ``good'' quasi-free states which well illustrate the theory. Such states in fact emerge from the actual geometry of the 
underlying Lorentzian manifold. We do, however, in section \ref{S9} we point out the existence of a structure, thoroughly described in our paper 
\cite{LM}, which characterizes states with well-defined expectation values on a large family of unbounded field observables.

We reiterate that the idea of quasi-free state, introduced by Derek Robinson \cite{Rob1}, was suggested by QFT where models of free and generalized free fields are characterized 
by the property that all their Wightman functions $W_n$ are uniquely determined by their two point functions $W_2$. This concept is also important in Quantum Statistical Physics, see 
\cite{BR2} for details.

We show how for any $\mathcal{O}\in\mathcal{K}(M,g)$ we may construct such a quasi-free state conditioned to $\mathcal{O}$. As before we shall follow the notational convention of \cite{BGP}. 
Let $\widetilde{G}_{\mathcal{O}}^\pm$ be the adv/ret Green's operators associated with $\mathcal{O}$ and $\widetilde{G}_{\mathcal{O}}= \widetilde{G}_{\mathcal{O}}^+-\widetilde{G}_{\mathcal{O}}^-$. 
A crucial ingredient in the construction of this example is the local boundedness of the Green's operator noted below:

\begin{proposition}\label{Gbd}\cite[Proposition 3.1]{FV}
For any $\mathcal{O}\in\mathcal{K}(M,g)$, the operator $i\widetilde{G}_\mathcal{O}$ is a densely defined symmetric operator on $L^2(\mathcal{O},dV)$ with domain $\mathscr{D}(\mathcal{O},E)$ 
which extends to a bounded operator $i\overline{G}_\mathcal{O}$ on $L^2(\mathcal{O},dV)$.
\end{proposition}

\begin{example}
Given $f,g\in \mathscr{D}(\mathcal{O},E)$, we will by the preceding proposition and Lemma 4.3.5 of \cite{BGP} have that 
\begin{eqnarray*}
\left|\int_M\langle f, \widetilde{G}_{\mathcal{O}}g\rangle\,dV\right| &=& \left|\int_M\langle f+\mathrm{ker}(\overline{G}_\mathcal{O}), \overline{G}_\mathcal{O}(g)\rangle\,dV\right| \\
&\leq& \|\overline{G}_\mathcal{O}\|.\|f+\mathrm{ker}(\overline{G}_\mathcal{O})\|.\|g+\mathrm{ker}(\overline{G}_\mathcal{O})\|\\
&\leq& \|\overline{G}_\mathcal{O}\|.\|f+\overline{\mathrm{ker}(\widetilde{G}_\mathcal{O})}\|.\|g+\overline{\mathrm{ker}(\widetilde{G}_\mathcal{O})}\|
\end{eqnarray*}
The final inequality follows from the fact that $\overline{\mathrm{ker}(\widetilde{G}_\mathcal{O})}\subset \mathrm{ker}(\overline{G}_\mathcal{O})$.

We will extract an inner product satisfying equation (5.30) of \cite{KM} from the above and use that to construct a quasi-free state. To do this we will make use of the correspondences described in 
\cite[Proposition 5.2.12]{KM}. In our setting the space $\mathrm{SOL}(\mathcal{O})$ referred to by \cite{KM} will just be $\mathrm{SOL}(\mathcal{O})= 
\{\widetilde{G}_{\mathcal{O}}(f)\colon f\in \mathscr{D}(\mathcal{O},E)\}$ and the space $\mathcal{E}(\mathcal{O})$ the quotient space $\mathcal{E}(\cO)= 
\mathscr{D}(\mathcal{O},E)/\mathrm{ker}(\widetilde{G}_{\mathcal{O}})$ (to see this consider equation (5.3) of \cite{KM} alongside the definition of this space). If one compares equations (5.4) and 
(5.12) in \cite{KM}, the the bilinear form $\tau$ they define on $\mathrm{SOL}$ is just $\tau(\widetilde{G}_{\mathcal{O}}(f), \widetilde{G}_{\mathcal{O}}(g)) =\int_M\langle f, 
\widetilde{G}_{\mathcal{O}}(g)\rangle\,dV$. Writing $[f]$ for the equivalence class $f+\overline{\mathrm{ker}(\widetilde{G}_\mathcal{O})}$, it is now an exercise to see that 
$\int_M\langle f, \widetilde{G}_{\mathcal{O}}(g)\rangle\,dV = \int_M\langle x_f, \widetilde{G}_{\mathcal{O}}(y_g)\rangle\,dV$ for any $x_f\in [f]$, $y_g\in [g]$. Similarly $\widetilde{G}_{\mathcal{O}}f 
= \overline{G}_{\mathcal{O}}[f]$. That means that we may write $\int_M\langle f, \widetilde{G}_{\mathcal{O}}(g)\rangle\,dV = \int_M\langle [f], \overline{G}_{\mathcal{O}}[g]\rangle\,dV$.  For any 
$f,g \in \mathscr{D}(\mathcal{O},E)$, it follows from the previously displayed inequality that $$|\tau(\widetilde{G}_\mathcal{O}(f), \widetilde{G}_{\mathcal{O}}(g))|^2 \leq 
\|\overline{G}_\mathcal{O}\|^2\|[f]\|^2\|[g]\|^2$$for all $[f], [g]\in \mathcal{E}(\mathcal{O})$. We write $H_{[\cdot]}$ for the real Hilbert space generated by $\mathcal{E}(\mathcal{O})$ equipped 
with the quotient norm norm inherited from $L^2(\mathcal{O},dV)/\mathrm{ker}(\overline{G}_\mathcal{O})$ Here we worked with real Hilbert spaces for the sake of simplicity, but all of the above can of 
course be complexified in the obvious manner, which we now assume to be the case. Thus the prescription 
$$\mu^\mathcal{O}([f],[g])=\tfrac{\|\overline{G}_\mathcal{O}\|}{2}\langle [f],[g]\rangle\quad f, g\in \mathscr{D}(\mathcal{O},E)$$ 
defines an inner product on $\mathcal{E}$ which will satisfy 
$$\frac{1}{4}|\tau(\widetilde{G}_{\mathcal{O}}(f), \widetilde{G}_{\mathcal{O}}(g))|^2\leq \mu^{\mathcal{O}}([f],[f]).\mu^{\mathcal{O}}([g],[g])\mbox{ for all }[f],[g]\in \mathcal{E}.$$
What we need in order to have access to \cite[Propositions 5.2.23 \& 5.2.24]{KM} is much the same thing, but for $\mathrm{SOL}(\mathcal{O})$ not $\mathcal{E}(\mathcal{O})$. Since these two spaces are 
bijective by \cite[Proposition 5.2.12]{KM} we may now define a matching inner product $\widetilde{\mu}^\mathcal{O}$ on $\mathrm{SOL}(\mathcal{O})$ by simply setting 
$\widetilde{\mu}^\mathcal{O}(\widetilde{G}_{\mathcal{O}}(f),\widetilde{G}_{\mathcal{O}}(g))=\mu^{\mathcal{O}}([f],[g])$ for all $[f], [g]\in \mathcal{E}$. That will then furnish us with an inner 
product which satisfies (5.30) of \cite{KM}, namely that 
\begin{equation}\label{KM:5.30}
\frac{1}{4}|\tau(x, y)|^2\leq \widetilde{\mu}^{\mathcal{O}}(x,x).\widetilde{\mu}^{\mathcal{O}}(y,y)\mbox{ for all }x,y \in \SOL(\mathcal{O}).
\end{equation}
By definition (see also Theorem 5.2.24 in \cite{KM}) there is then a unique quasi-free state $\omega_\mathcal{O}$ corresponding to the two-point functions 
$(\omega_\mathcal{O})_2(f,g)=\widetilde{\mu}^{\mathcal{O}}(\widetilde{G}_{\mathcal{O}}f,\widetilde{G}_{\mathcal{O}}g)+i\frac{1}{2}\tau(\widetilde{G}_{\mathcal{O}}(f), \widetilde{G}_{\mathcal{O}}(g))$. 
This follows, see \cite{MV}, from the recipe for a quasi-free state described there, namely 
\begin{equation}
\omega_s(W(x)) = \exp\{-\tfrac{1}{2}s(x,x)\}
\end{equation}
where the quadratic form $s(\cdot, \cdot)$ satisfies $|\sigma(f,g)|^2 \leq s(f,f) s(g,g)$, $\sigma$ is the symplectic form, and there is the corresponding formula for the two-point function. 
We shall write $\cM(\mathcal{O})$ for the resultant von Neumann algebra.
\end{example}

\subsection{Constructing Araki-Woods factors}\label{AWF}

We provide further details on the construction of Araki-Woods factors in the present setting. Our starting point is \cite[Proposition 5.2.23 \& Theorem 5.2.24]{KM}. Our exposition 
makes use of details of the proof of \cite[Proposition 3.1]{KW}. Let $\mathcal{O}\in \mathcal{K}(M,g)$ and consider the quasi-free state on $\cM(\mathcal{O})$ with the two point function 
given by $$(\widetilde{G}_\mathcal{O}f,\widetilde{G}_\mathcal{O}g)\to \widetilde{\mu}^{\mathcal{O}}(\widetilde{G}_\mathcal{O}f,\widetilde{G}_\mathcal{O}g)+ 
i\frac{1}{2}\tau(\widetilde{G}_\mathcal{O}f, \widetilde{G}_\mathcal{O}g)$$ where $f,g \in\mathscr{D}(\mathcal{O},E)$. The completion of the vector space $\SOL(\mathcal{O})$ under the inner product 
$\widetilde{\mu}^{\mathcal{O}}$ then yields a real Hilbert space $H_\mu$. Our task is to turn this real Hilbert space into a complex Hilbert space. The first step in this endeavour is to 
notice that an application of the Riesz Lemma to equation (\ref{KM:5.30}) ensures there exists a contractive operator $A_\mathcal{O}$ on $H_\mu$ satisfying $A_\mathcal{O}^*=-A_\mathcal{O}$ 
such that $$\frac{1}{2}\tau(x,y)=\widetilde{\mu}^{\mathcal{O}}(x,A_\mathcal{O}y)\mbox{ for all }x,y\in \SOL(\mathcal{O}).$$In the present setting the non-degeneracy of $\tau$ ensures that 
$A_\mathcal{O}$ is injective. The isometry $\mathcal{J}_{\cO}$ in the polar decomposition of $A_\mathcal{O}$ can then be shown to satisfy $[\mathcal{J}_{\cO},A_\mathcal{O}]=0$, 
$\mathcal{J}_{\cO}^2=-\I$ and $\mathcal{J}_{\cO}^*=-\mathcal{J}_{\cO}$. $H_\mu$ may now be made 
into a complex Hilbert space $\mathcal{H}_\mu$ by for all $x,y\in \SOL(\mathcal{O})$ defining a new inner product as follows:
\begin{align*}
\langle\!\langle x, y\rangle\!\rangle_\mathcal{O} &= \widetilde{\mu}^{\mathcal{O}}(x,y) +i\widetilde{\mu}^{\mathcal{O}}(x,\mathcal{J}_{\cO}y)\\
ix &= -\mathcal{J}_{\cO}x
\end{align*}
We next define an operator $K_\mathcal{O}:\SOL(\mathcal{O})\to \mathcal{H}_\mu\oplus \mathcal{H}_\mu$ by means of the prescription 
$$K_\mathcal{O}(x) = [\tfrac{1}{2}(|A_\mathcal{O}|+\I)]^{1/2}x\oplus C[\tfrac{1}{2}(\I-|A_\mathcal{O}|)]^{1/2}x$$where $C$ is complex conjugation. One may therefore equivalently define the operator 
$K_\mathcal{O}$ as a map into $\mathcal{H}_\mu\oplus \overline{\mathcal{H}}_\mu$, where $\overline{\mathcal{H}}_\mu$ is the conjugate Hilbert space, by means of the prescription 
$K_\mathcal{O}(x) = [\frac{1}{2}(|A_\mathcal{O}|+\I)]^{1/2}x\oplus \overline{[\frac{1}{2}(\I-|A_\mathcal{O}|)]^{1/2}x}$. One then defines $H_\mathcal{O}$ to be 
the closure of $\mathrm{ran}(K_\mathcal{O})+i\mathrm{ran}(K_\mathcal{O})$ in $\mathcal{H}_\mu\oplus \overline{\mathcal{H}}_\mu$; equivalently as 
$\mathcal{H}_\mu\oplus \chi_{\mathbb{R}\backslash\{1\}}(\overline{|A_\mathcal{O}|})\overline{\mathcal{H}}_\mu$. The pair $(H_\mathcal{O}, K_{\mathcal{O}})$ is the 1-particle 
structure described in \cite[Proposition 5.2.23]{KM}. With $\omega_\mathcal{O}$ denoting the resultant quasi-free state, a precise description of the GNS space generated by $\omega_\mathcal{O}$ 
may then be found in \cite[Theorem 5.2.24(b)]{KM}.

The approach to construct the GNS space for $\omega_\mathcal{O}$ described above, is the Kay-Wald approach. The Araki-Woods approach is similar 
and yields a copy of the above GNS space. However it uses a different approach to construct a complex inner product. We shall refer to \cite[\S 4.8.2]{Ger} to elucidate this 
approach. Specifically with $A_\mathcal{O}$ and $\mathcal{J}_{\cO}$ as before this approach uses the complex inner product on $\SOL(\mathcal{O})$ given by 
$$\langle x, y\rangle_{AW} = -\tau(x,\mathcal{J}_{\cO}y)+i\tau(x,y).$$(The sign differences in the formula recorded here when compared with that in \cite{Ger} is for the purpose of aligning that 
two exposition with that of \cite{KW}.) The two inner products are related in the sense that for all $x, y\in \SOL(\mathcal{O})$ one has  
\begin{eqnarray*}
\langle x, y\rangle_{AW} &=& -\tau(x,\mathcal{J}_{\cO}y)+i\tau(x,y)\\
&=& -2\widetilde{\mu}^{\cO}(x,A_\mathcal{O}\mathcal{J}_{\cO}y)-i\tau(x,\mathcal{J}_{\cO}^2y)\\
&=& 2\widetilde{\mu}^{\cO}(x,|A_\mathcal{O}|y)-i2\widetilde{\mu}^{\cO}(x,\mathcal{J}_{\cO}^2A_\mathcal{O}y)\\
&=& 2[\widetilde{\mu}^{\cO}(x,|A_\mathcal{O}|y)+i\widetilde{\mu}^{\cO}(x,\mathcal{J}_{\cO}|A_\mathcal{O}|y)]\\
&=& 2\langle\!\langle x, |A_\mathcal{O}|y\rangle\!\rangle_\mathcal{O}
\end{eqnarray*}
Writing $\mathcal{H}_\mu^{AW}$ for the complex Hilbert space yielded by the above inner product, the Araki-Woods version of the 1-particle space is given by 
$H_\mathcal{O}^{AW}=\mathcal{H}_\mu^{AW}\oplus \chi_{\mathbb{R}\backslash\{1\}}(\overline{|A_\mathcal{O}|})\overline{\mathcal{H}}_\mu^{AW}$. We leave it as an exercise to verify that 
the operator $\frac{1}{\sqrt{2}}|A_\mathcal{O}|^{-1/2}$ induces a densely defined isometry with dense range from $\mathcal{H}_{\mu}$ to $\mathcal{H}_\mu^{AW}$. These two Hilbert spaces 
are therefore unitarily equivalent. It follows from previously displayed set of equalities that for the operator $\rho=\frac{1}{2}|A_\mathcal{O}|^{-1}$ 
we have that $\langle x, \rho x\rangle_{AW} = \widetilde{\mu}^{\mathcal{O}}(x,x)$ for all $x\in \SOL(\mathcal{O})$. This is exactly 
the operator $\rho$ required to construct the Araki-Woods factors as described in \cite[\S 9.3]{Der}. The operator $\gamma$ which in Derezinski's description of the Araki-Woods 
approach parametrises the construction of these factors (see \cite[\S 9.3]{Der}), then corresponds to $\gamma = \rho(\I+\rho)^{-1} = (\I+2|A_\mathcal{O}|)^{-1}$. For details of the above 
facts we refer the reader to \cite[\S 4.8.2]{Ger}.

In view of the above correspondence we shall where convenient pass to the Araki-Woods representation to prove some of the facts 
required subsequently. 

\subsection{The behaviour of $\{\cM(\mathcal{O}):\mathcal{O}\in\mathcal{K}(M,g)\}$ with respect to local isometries}

Let $\mathcal{O}\in \mathcal{K}(M,g)$ be as before and let $\Psi$ be an isometry which preserves time-orientation. 
Consider the quasi-free state on $\cM(\mathcal{O})$ with the two point function given by $(\widetilde{G}_\mathcal{O}f,\widetilde{G}_\mathcal{O}g)\to 
\widetilde{\mu}^{\mathcal{O}}(\widetilde{G}_\mathcal{O}f,\widetilde{G}_\mathcal{O}g)+i\frac{1}{2}\tau(\widetilde{G}_\mathcal{O}f, \widetilde{G}_\mathcal{O}g)$ where $f,g \in\mathscr{D}(\mathcal{O},E)$. 
We know from Theorem \ref{localg-ex} that for $f, g\in \mathscr{D}(\mathcal{O},E)$ we have 
\begin{eqnarray*}
\tau(\widetilde{G}_{\mathcal{O}}(g), \widetilde{G}_{\mathcal{O}}(f)) &=& \int_M\langle g, \widetilde{G}_{\mathcal{O}}(f)\rangle\,dV\\
&=& \int_M\langle g\circ\Psi^{-1}, \widetilde{G}_{\Psi(\mathcal{O})}(f\circ\Psi^{-1})\rangle\,dV\\
&=& \tau(\widetilde{G}_{\Psi(\mathcal{O})}(g\circ\Psi^{-1}), \widetilde{G}_{\Psi(\mathcal{O})}(f\circ\Psi^{-1})).
\end{eqnarray*}
We show that we similarly have that 
$$\widetilde{\mu}^{\mathcal{O}}(\widetilde{G}_\mathcal{O}f,\widetilde{G}_\mathcal{O}g) = \widetilde{\mu}^{\Psi(\mathcal{O})}(\widetilde{G}_{\Psi(\mathcal{O})}(f\circ\Psi^{-1}),
\widetilde{G}_{\Psi(\mathcal{O})}(g\circ\Psi^{-1})).$$To see this we firstly note that the facts noted about the preservation of measure in the first part of the proof of Theorem
\ref{localg-ex}, ensure that $$\|f\|_\mathcal{O}^2=\int_M\langle f, f\rangle\,dV = \int_M\langle f\circ\Psi^{-1}, f\circ\Psi^{-1}\rangle\,dV= \|f\circ\Psi^{-1}\|^2_{\Psi(\mathcal{O})}$$for all 
$f,g\in \mathscr{D}(\mathcal{O},E)$. When combined with Lemma \ref{Green}, this in fact ensures that that $\|\overline{G}_\mathcal{O}\|= \|\overline{G}_{\Psi(\mathcal{O})}\|$. The claimed 
equality will then follow if we can additionally show that $\langle[f],[g]\rangle = \langle[f\circ\Psi^{-1}],[g\circ\Psi^{-1}]\rangle$ 
for all $f,g\in \mathscr{D}(\mathcal{O},E)$. This in turn firstly follows from the fact noted above that $f\to f\circ\Psi^{-1}$ defines a 
unitary from $L^2(\mathcal{O},dV)$ to $L^2(\Psi(\mathcal{O}),dV)$, and secondly from the fact that by Lemma \ref{Green} this prescription 
also bijectively maps $\mathrm{ker}(\widetilde{G}_\mathcal{O})$ onto $\mathrm{ker}(\widetilde{G}_{\Psi(\mathcal{O})})$. We therefore have 
that $\|[f]\|_\mathcal{O}^2= \|[f\circ\Psi^{-1}]\|^2_{\Psi(\mathcal{O})}$, as was required. We equivalently have that 
$$\|\widetilde{G}_{\mathcal{O}}(f)\|_{\widetilde{\mu}_\mathcal{O}}= \|\widetilde{G}_{\Psi(\mathcal{O})}(f\circ\Psi^{-1}]\|_{\widetilde{\mu}_{\Psi(\mathcal{O})}}
\mbox{for all }[f]\in\mathcal{E}(\mathcal{O}).$$

We proceed to sum up the above analysis. The context is set by the symplectic space, $(\SOL(\mathcal{O}), \tau_{\cO}) = \mathscr{D}(\cO, E)/\mathrm{ker}(\widetilde{G}_\mathcal{O})$. (See 
\cite[Page 129]{BGP}.) Within this context we observed the following.
\begin{enumerate}
\item The map $T: (\SOL(\mathcal{O}), \tau_{\cO}) \to (\SOL(\mathcal{O}, \tau_{\Psi(\cO)}): \widetilde{G}_\mathcal{O}f \mapsto \widetilde{G}_{\Psi(\mathcal{O})}(f \circ \Psi^{-1})$ is a symplectic map, 
for the $\mathrm{CCR}$ algebras $(\cA(\cO), \tau_{\cO})$ and  $(\cA(\Psi(\cO)), \tau_{\Psi(\cO)})$ where $\cA(\cO)$ is the Weyl algebra over the symplectic space $(\SOL(\mathcal{O}),\tau_{\cO})$, etc. 
The Bogoliubov transform therefore ensures that 
the given prescription induces a *-isomorphism $\alpha_\Psi$ from $(\cA(\cO), \tau_{\cO})$ onto  $(\cA(\Psi(\cO)), \tau_{\Psi(\cO}))$.
\item On $\cA((\cO))$ and $\cA(\Psi(\cO))$ there exist quasi-free states $\omega_{\cO}$ and $\omega_{\Psi(\cO)}$ with the corresponding 
two-point functions respectively given by 
$$(\widetilde{G}_\mathcal{O}f,\widetilde{G}_\mathcal{O}g)\to \widetilde{\mu}^{\mathcal{O}}(\widetilde{G}_\mathcal{O}f,\widetilde{G}_\mathcal{O}g)+i\frac{1}{2}\tau(\widetilde{G}_\mathcal{O}f, 
\widetilde{G}_\mathcal{O}g)$$ and 
\begin{align*}
&(\widetilde{G}_{\Psi(\mathcal{O})}(f\circ\Psi^{-1}),\widetilde{G}_{\Psi(\mathcal{O})}(g\circ\Psi^{-1})) \\
& \to \qquad \widetilde{\mu}^{\Psi(\mathcal{O})}(\widetilde{G}_{\Psi(\mathcal{O})}(f\circ\Psi^{-1}),\widetilde{G}_{\Psi(\mathcal{O})}(g\circ\Psi^{-1}))\\
& \qquad \qquad +i\frac{1}{2}\tau(\widetilde{G}_{\Psi(\mathcal{O})}(f\circ\Psi^{-1}), \widetilde{G}_{\Psi(\mathcal{O})}(g\circ\Psi^{-1}))
\end{align*}
where $f,g\in \mathscr{D}(\mathcal{O},E)$.
\item The self-same map $\widetilde{G}_\mathcal{O}f \mapsto \widetilde{G}_{\Psi(\mathcal{O})}(f \circ \Psi^{-1})$ precisely identifies the tools needed to construct $\omega_{\cO}$ with 
those needed to construct $\omega_{\Psi(\cO)}$. When these facts are considered alongside the process described in subsection \ref{AWF}, it is clear that we also have that 
$\omega_{\Psi(\cO)}\circ \alpha_\Psi = \omega_{\cO}$. 
\item It therefore follows from the discussion in the third paragraph of section \ref{S3} that the standard forms of $\pi_{\omega_{\cO}}(\mathcal{A}(\cO))''$ and 
$\pi_{\omega_{\Psi(\cO)}}(\mathcal{A}(\Psi(\cO)))''$ are isomorphic.
\end{enumerate}

We therefore arrive at the following:

\begin{proposition}
$\alpha_\Psi$ has a weak extension $(\alpha_\psi)_{ext}$ which yields a $^*$-isomorphism between $(\cM(\cO),\omega_{\cO})$ and $(\cM(\Psi(\cO)),\omega_{\Psi(\cO)})$. This *-isomorphism 
yields a spatial isomorphism of the form described in Theorem \ref{T8.haag-stdfm} which identifies the standard forms of the von Neumann algebras $\cM(\cO)$ and $\cM(\Psi(\cO))$ as respectively 
determined by $\omega_{\cO}$ and $\omega_{\Psi(\cO)}$. There are therefore well-defined mappings between the nets of von Neumann measure algebra $\cO \mapsto (\cM(\cO),\omega_{\cO})$ and  
$\Psi(\cO) \mapsto (\cM(\Psi(\cO)),\omega_{\Psi(\cO)})$ where $\cO \in \cK(M,g)$.
\end{proposition}

\subsection{Isotonic properties of $\{\cM(\mathcal{O}):\mathcal{O}\in\mathcal{K}(M,g)\}$}

Given $\mathcal{O}$ we will write $(\cM(\cO), \mathcal{F}^+(H_\mathcal{O}), J_\mathcal{O}, \mathscr{P}_\mathcal{O})$ for the standard form of $\cM(\cO)$ determined by the GNS 
representation generated by $\omega_{\mathcal{O}}$. (Here we made use of the fact that the GNS Hilbert space is precisely $\mathcal{F}^+(H_\mathcal{O})$ 
\cite[Theorem 5.2.24(b)]{KM}.) Let $\mathcal{O}_1, \mathcal{O}_2\in\mathcal{K}(M,g)$ be given with $\mathcal{O}_1\subset \mathcal{O}_2$. We will show that 
there is a standard form homomorphism (in the sense of definition \ref{Stdfm-cat}) from the GNS standard form of the pair $(\cM(\cO_1),\omega_{\mathcal{O}_1})$ into the GNS 
standard form of the pair $(\cM(\cO_2),\omega_{\mathcal{O}_2})$. We shall achieve this by constructing a subalgebra of $\cM(\cO_2)$ whose GNS standard form is 
isomorphic to that of the pair $(\cM(\cO_1),\omega_{\mathcal{O}_1})$ and which also embeds into that of $(\cM(\cO_2),\omega_{\mathcal{O}_2})$. 
 
For ease of notation we will hereafter where convenient write $\mu_k$, $\omega_k$, $J_k$, $G_k$ \dots for $\widetilde{\mu}^{\mathcal{O}_k}$, $\omega_{\mathcal{O}_k}$, $J_{\mathcal{O}_k}$, 
$\widetilde{G}_{\cO_k}(f_{ext})$,\dots. At the risk of a clash of notation we will denote the $L^2$-closure of a subspace $F$ of $L^2(\mathcal{O}_k,dV)$ by $[F]$. Here $H_{\mu_k}$ is of course 
then a scaled version of $H_{[\cdot]_k}=[\mathscr{D}(\mathcal{O},E)]/[\mathrm{ker}(\widetilde{G}_k)]$. To complete the task we set for ourselves we need some information 
on how the Klein-Gordon PDO behaves with respect to a ``shift in focus'' from $\mathcal{O}_1$ to $\mathcal{O}_2$. For this we need the following proposition. In proving this 
proposition we shall constantly make use of the fact that the PDO $P_\mathcal{M}$ corresponding to the Klein-Gordon equation on $\mathcal{M}$ will map $C^\infty(\mathcal{O},E)$ 
back $C^\infty(\mathcal{O},E)$ and that the operator $P_{\mathcal{O}}$ is nothing but the restriction of $P_\mathcal{M}$ to the $\mathcal{O}$ context. (See \cite[Proposition 3.5.1]{BGP} and the 
discussion surrounding it.) 

\begin{proposition} Let $\mathcal{O}\in\mathcal{K}(M,g)$ and $f\in C^\infty(\mathcal{O},E)$ be given. If $f$ has bounded partial derivatives of all orders, then 
there exists $(f_n)\subset \mathscr{D}(\mathcal{O},E)$ such that $P_\mathcal{O}f_n\to P_\mathcal{O}f$ in $L^2(\mathcal{O},dV)$.
\end{proposition}

In the proof we shall make use of Sobolev space techniques. For background on Sobolev spaces we refer the reader to \cite{Evans}.

\begin{proof} \emph{Step 1:} We will first prove the proposition for the case where the manifold in question is $\mathbb{R}^N$. We prove that there exists a 
sequence of functions $(f_n)\subset \mathscr{D}(\mathcal{O},E)$ which converges to $f$ in the Sobolev space $W^{2,2}(\mathcal{O})$. Since on $\mathbb{R}^N$ 
the Klein-Gordon operator is a PDO with a very regular structure, it will then follow from this fact and the theory of Sobolev spaces that we must then have that 
$P_\mathcal{O}f_n\to P_\mathcal{O}f$ in $L^2(\mathcal{O},d\lambda_N)$.

We firstly note that the assumption on $f$ ensures that we will for any multi-index $\alpha$ with $|\alpha|\leq 2$ have that 
$\int_{\mathcal{O}}\langle D^\alpha(f),D^\alpha(f)\rangle\,d\lambda_N \leq \|D^\alpha(f)\|_\infty^2 \lambda_N(\mathcal{O})<\infty$ (where $\lambda_N$ is Lebesgue 
measure). So $f$ clearly belongs to the Sobolev space $W^{2,2}(\mathcal{O})$. We now introduce the open subsets $\mathcal{O}_k= 
\{x\in \mathcal{O}:d(x,\partial\mathcal{O})\leq 2^{-k}\}$. We clearly have that $\mathcal{O}_k\subset \mathcal{O}_{k+1}$ for all $k$ with $\mathcal{O} = 
\cup_{k=1}\mathcal{O}_k$. A similar argument to what we used above now shows that we will here for any multi-index $\alpha$ with $|\alpha|\leq 2$ have that 
$\int_{\mathcal{O}\backslash\mathcal{O}_k}\langle D^\alpha(f),D^\alpha(f)\rangle\,d\lambda_N \leq \|D^\alpha(f)\|_\infty^2 
\lambda_N(\mathcal{O}\backslash\mathcal{O}_k)\to 0$. That ensures that for the Sobolev space $W^{2,2}(\mathcal{O})$ we have that 
$\|f\chi_{\mathcal{O}\backslash\mathcal{O}_k}\|_{2,2}\to 0$ as $k\to \infty$.

We may also for any $k$ select a bump function $\zeta_k$ such that $\zeta_k=1$ on $\overline{\mathcal{O}_k}$ and $\zeta_k=0$ on 
$\mathcal{O}\backslash\overline{\mathcal{O}_{k+1}}$ (see the corollary to \cite[Theorem 1.11]{War}). The function $\zeta_kf$ will then be a function of compact 
support which agrees with $f$ on $\mathcal{O}_k$ and which is supported on $\mathcal{O}_{k+1}$. We therefore have that 
$$f-\zeta_k.f= f\chi_{\mathcal{O}\backslash\mathcal{O}_k} -\zeta_k.f\chi_{\mathcal{O}_{k+1}\backslash\mathcal{O}_k}.$$We already know that
$\|f\chi_{\mathcal{O}\backslash\mathcal{O}_k}\|_{2,2}\to 0$ as $k\to \infty$. So if we can show that inside $W^{2,2}(\mathcal{O})$ the function 
$\zeta_k.f\chi_{\mathcal{O}_{k+1}\backslash\mathcal{O}_k}$ is the limit of a sequence of smooth functions of compact support, it will 
therefore follow that inside $W^{2,2}(\mathcal{O})$, the function $f$ is then similarly the limit of a sequence of smooth functions of 
compact support. To prove that this the case, we use the theory of mollifiers (see subsection C.4 in \cite{Evans}). Given such a 
mollifier $\eta_\epsilon$ it is an exercise to see that we will for any locally integrable function $g$ on $\mathcal{O}$ have that 
$(\eta_\epsilon*g)(x)= \int_{B(0,\epsilon)}\eta_\epsilon(y)g(x-y)\,dy =0$ whenever $d(x,\mathrm{supp}(g))>\epsilon$. For any 
$\epsilon\geq 2^{-k-2}$ the function $\eta_\epsilon*(\zeta_k.f\chi_{\mathcal{O}_{k+1}\backslash\mathcal{O}_k})$ will therefore be 
supported on $\overline{\mathcal{O}_{k+2}}\backslash\mathcal{V}_k$ where $\mathcal{V}_k=\{x\in\mathcal{O}_k:d(x,\partial\mathcal{O}_k)>2^{-k-2}\}$. 

Assuming this to be the case, the functions $\eta_\epsilon*(\zeta_k.f\chi_{\mathcal{O}_{k+1}\backslash\mathcal{O}_k})$ are then all smooth 
functions of compact support \cite[\S C.4, Theorem 6]{Evans} for which the restrictions to any open subset $V$ with 
$\overline{V}\subset\mathcal{O}$ will converge to $\zeta_k.f\chi_{\mathcal{O}_{k+1}\backslash\mathcal{O}_k}$ inside $W^{2,2}(V)$ as 
$\epsilon\to 0$ \cite[\S 5.3.1, Theorem 1]{Evans}. But $\mathcal{O}_{k+2}$ is such a subset and each 
$\eta_\epsilon *(\zeta_k.f\chi_{\mathcal{O}_{k+1}\backslash\mathcal{O}_k})$ as well as 
$\zeta_k.f\chi_{\mathcal{O}_{k+1}\backslash\mathcal{O}_k}$ are supported on $\mathcal{O}_{k+2}$. Hence as $\epsilon\searrow 0$, the 
functions $\eta_\epsilon*(\zeta_k.f\chi_{\mathcal{O}_{k+1}\backslash\mathcal{O}_k})$ even converge to 
$\zeta_k.f\chi_{\mathcal{O}_{k+1}\backslash\mathcal{O}_k}$ in $W^{2,2}(\mathcal{O})$. This then proves the proposition in the $\mathbb{R}^N$ case.

\bigskip

\emph{Step 2:} The manifold $(M,g)$ admits an orientable atlas, and so by compactness $\overline{\mathcal{O}}$ can be covered with finitely many oriented 
charts $U_1,\dots, U_m$ (see \cite[Definition 1.41]{ONeil}) which correspond to relatively compact open regions in $\mathbb{R}^N$ via diffeomorphisms 
$\rho_i$ from $V_i$ to $U_i$. We can also arrange matters so that the Jacobian $|\mathrm{det}(d\rho_i)|$ is bounded. (To see this note that for such triple 
$(U,V,\rho)$ where $U$ is a neighbourhood of $p\in \mathcal{O}$ we can first replace $U$ with a smaller compact neighbourhood $K$ (ensuring the compactness 
of $\rho_i^{-1}(K)$ and the boundedness of the Jacobian, and then replace $U$ with the interior of $K$.) The compact set $\overline{\mathcal{O}}$ can then 
be covered with such modified charts, with compactness ensuring that only finitely many such charts are necessary. For any given $U_i$ we will only work 
with $U_i\cap\mathcal{O}$. In so doing we will wherever convenient assume that functions supported on $U_i\cap\mathcal{O}$ are extended to all of $U_i$ by 
assigning the value 0 on $U_i\backslash(U_i\cap\mathcal{O})$ (and similarly for functions supported on $\rho_i^{-1}(U_i\cap\mathcal{O})$). Since 
$f\in C^\infty(\mathcal{O},E)$ is a function with partial derivatives of all orders bounded, the function $f{\upharpoonright}{U_i\cap\mathcal{O}}\circ\rho_i$ is similarly 
a smooth function with partial derivatives of all orders bounded. 
Writing $\widetilde{P}_i$ for the coordinatisation of $P_{U_i}$, it follows from step 1 that there exists a sequence 
$(g_{n,i})\subset \mathscr{D}(U_i\cap\mathcal{O},E)$ for which $(\widetilde{P}_i(g_{n,i}))$ converges to $\widetilde{P}_i(f{\upharpoonright}{U_i\cap\mathcal{O}}\circ\rho_i)$ 
in $L^2(\rho_i^{-1}(U_i\cap\mathcal{O}),\lambda_N)$-norm. The diffeomorphisms $\rho_i$ now induce measurable isomorphisms between the Borel algebras 
$\mathcal{B}(U_i)$ and $\mathcal{B}(V_i)$ with the Jacobian playing the role of a Radon-Nikodym derivative (see \cite[Prop 5.19]{GQ2}). We therefore have that 
\begin{eqnarray*}
&&\int_{U_i\cap\mathcal{O}}\langle P_i(f{\upharpoonright}{U_i\cap\mathcal{O}})-P_i(g_{n,i}\circ\rho_i^{-1}), P_i(f{\upharpoonright}{U_i\cap\mathcal{O}})-P_i(g_{n,i}\circ\rho_i^{-1})\rangle\,dV\\
&=& \int|\langle \widetilde{P}_i(f{\upharpoonright}{U_i\cap\mathcal{O}}\circ\rho_i)-\widetilde{P}_i(g_{n,i}), 
\widetilde{P}_i(f{\upharpoonright}{U_i\cap\mathcal{O}}\circ\rho_i)-\widetilde{P}_i(g_{n,i})\rangle|\,|\mathrm{det}(d\rho_i)|\,d\lambda_N\\
&\leq& \|\,|\mathrm{det}(d\rho_i)|\,\|_\infty\int|\langle \widetilde{P}_i(f{\upharpoonright}{U_i\cap\mathcal{O}}\circ\rho_i)-\widetilde{P}_i(g_{n_i}), 
\widetilde{P}_i(f{\upharpoonright}{U_i\cap\mathcal{O}}\circ\rho_i)-\widetilde{P}_i(g_{n,i})\rangle\,d\lambda_N
\end{eqnarray*}
Thus $(P_i(g_{n,i}\circ\rho_i^{-1}))$ converges to $P_i(f{\upharpoonright}{U_i\cap\mathcal{O}})= P_{\mathcal{O}}(f){\upharpoonright}{U_i\cap\mathcal{O}}$. Since $g_{n,i}\in 
\mathscr{D}(\rho_i^{-1}(U_i\cap\mathcal{O}), E)$ for each $n$, we clearly have that $g_{n,i}\circ\rho_i^{-1}\in \mathscr{D}(U_i\cap\mathcal{O}, E)$ for 
each $n$. We may then extend $g_{n,i}\circ\rho_i^{-1}$ to an element of $\mathscr{D}(\mathcal{O}, E)$ in the obvious way by assigning the value 0 to the 
extension on $\mathcal{O}\backslash U_i$. It is clear that then $P_\mathcal{O}(g_{n,i}\circ\rho_i^{-1})$ is then a similar extension of 
$P_i(g_{n,i}\circ\rho_i^{-1})$. The sequence of functions $(f_n)\subset \mathscr{D}(\mathcal{O}, E)$ defined by $f_n =\sum_{i=1}^m g_{n,i}\circ\rho_i^{-1}$ is 
then a sequence for which 
\begin{eqnarray*}
&&\int_M\langle P_\mathcal{O}(f-f_n), P_\mathcal{O}(f-f_n)\rangle\,dV \\
&\leq& \sum_{i=1}^m \int_{U_i\cap\mathcal{O}}\langle P_i(f{\upharpoonright}{U_i\cap\mathcal{O}})-P_i(g_{n,i}\circ\rho_i^{-1}), 
P_i(f{\upharpoonright}{U_i\cap\mathcal{O}})-P_i(g_{n,i}\circ\rho_i^{-1})\rangle\,dV\\
&\to& 0\mbox{ as }n\to\infty
\end{eqnarray*}     
\end{proof}

\begin{corollary}
For any $f\in \mathrm{ker}(\widetilde{G}_2)$, the restriction $f{\upharpoonright}{\mathcal{O}_1}$ will belong to $[\mathrm{ker}(\widetilde{G}_1)]$.
\end{corollary}

\begin{proof} We remind the reader that $\mathrm{ker}(\widetilde{G}_k)=P_{\mathcal{O}_k}(\mathscr{D}(\mathcal{O}_k,E))$. (See page 129 of \cite{BGP}.) So let 
$f=P_{\mathcal{O}_2}(g)$, where $g\in \mathscr{D}(\mathcal{O}_2,E)$, be given. (Here $P_{\mathcal{O}}$ is the partial differential operator corresponding to 
the Klein-Gordon equation on $\mathcal{O}$.) The corollary will then follows from the preceding proposition once we notice that since it is the restriction of 
a function of compact support, the function $g{\upharpoonright}{\mathcal{O}_1}\in C^\infty(\mathcal{O}_1,E)$ must have bounded partial derivatives of all orders.
\end{proof}

\begin{proposition}\label{[]isom} Let $e_{(2:1)}$ be the projection from $H_{[\cdot]_2}$ onto the closure of $\{[f_{ext}]_2:f\in\mathscr{D}(\mathcal{O}_1,E)\}$. Then the map 
$$U:f_{ext} + \overline{\mathrm{ker}(\widetilde{G}_2)} \mapsto f+\overline{\mathrm{ker}(\widetilde{G}_1)}\quad f\in \mathscr{D}(\mathcal{O}_1,E)$$extends to a unitary operator from 
$e_{(2:1)}H_{[\cdot]_2}$ onto $H_{[\cdot]_1}$.
\end{proposition}

\begin{proof}
Given $f\in\mathscr{D}(\mathcal{O}_1,E)$ with $\widetilde{G}_1(f)\neq 0$, let $b_f\in [\mathrm{ker}(\widetilde{G}_1)]$ be the unique element for which $f-b_f \perp [\mathrm{ker}(\widetilde{G}_1)]$. 
Select $(b_n)\subset \mathscr{D}(\mathcal{O}_1,E)$ such that $P_{\mathcal{O}_1}(b_n)\to b_f$. (Here we are using the fact noted earlier that $\mathrm{ker}(\widetilde{G}_k)= 
P_{\mathcal{O}_k}(\mathscr{D}(\mathcal{O}_k,E))$.) Since (again as noted earlier)  $P_{\mathcal{O}_2}(b_n^{ext})= P_{\mathcal{O}_1}(b_n)_{ext}$ for all $n$, we have that 
$(P_{\mathcal{O}_1}(b_n)_{ext}) \subset P_{\mathcal{O}_2}(\mathscr{D}(M_{\mathcal{O}_2},E))=\mathrm{ker}(\widetilde{G}_2)$. The fact that $\lim_n P_{\mathcal{O}_1}(b_n)_{ext}=b_f^{ext}$, therefore 
ensures that $b_f^{ext}\in [\mathrm{ker}(\widetilde{G}_2)]$. Given any $\upsilon \in  \mathrm{ker}(\widetilde{G}_2)$, we may now use the preceding corollary and the fact that $f$ and $b_f$ are 
supported on $\mathcal{O}_1$, to see that $$\int_M\langle (f-b_f)_{ext},\upsilon\rangle,dV = \int_M\langle f-b_f,\upsilon{\upharpoonright}{\mathcal{O}_1}\rangle,dV=0.$$Thus $(b_f)_{ext}$ is 
the unique element of $[\mathrm{ker}(\widetilde{G}_2)]$ for which $(f-b_f)_{ext}\perp [\mathrm{ker}(\widetilde{G}_2)]$. But that means that 
$\|f+[\mathrm{ker}(\widetilde{G}_1)]\|=\|f-b_f\|=\|f_{ext}+[P_{\mathcal{O}_2}(C_0^\infty(M_{\mathcal{O}_2},E)]]_2\|$. Thus the canonical map from 
$[ext(\mathscr{D}(\mathcal{O}_1,E))+\mathrm{ker}(\widetilde{G}_2]/[\mathrm{ker}(\widetilde{G}_2)]$ to 
$[\mathscr{D}(\mathcal{O}_1,E))]/[\mathrm{ker}(\widetilde{G}_1)]$ is in fact a surjective linear isometry.
\end{proof}

In the following we shall write $\SOL(\cO_2:\cO_1)$ for the subset $\{\widetilde{G}_{\cO_2}(f_{ext}): f\in \mathscr{D}(\cO_1,E)\}$ and similarly write $\cM(\cO_2:\cO_1)$ for the von Neumann subalgebra 
generated by the Weyl operators $\{\pi_2(W(x))\colon x\in \SOL(\cO_2:\cO_1)\}$. By suitably rescaling the above unitary operator, we now arrive at the following 

\begin{corollary} There is a unitary operator $v_{(2:1)}$ from $e_{(2:1)}H_{\mu_2}$ to $H_{\mu_1}$ which maps each $\widetilde{G}_{\cO_2}(f_{ext})$, where $f\in \mathscr{D}(\mathcal{O}_1,E)$, to 
$\widetilde{G}_{\cO_1}(f)$.  
\end{corollary} 

With the above technology now at our disposal we pass to an analysis of the von Neumann subalgebra of $\cM(\cO_2:\cO_1)$. Given $f, g \in \mathscr{D}(\cO_1,E)$, it follows fairly immediately 
that
\begin{eqnarray*}
\tau_2(G_2(f_{ext}), G_2(g_{ext})) &=& \int_{\cO_2}\langle f_{ext},G_2(g_{ext})\rangle\,dV\\
&=& \int_{\cO_1}\langle f, G_1(g)\rangle\,dV\\
&=& \tau_2(G_1(f), G_1(g))
\end{eqnarray*}
The above is a fairly easy consequence of the definition of $G_k$ as described in \cite[Proposition 3.5.1]{BGP} and the fact that here  
$g$ is supported on $\mathcal{O}_1$. The map $T: \SOL(\cO_1) \to \SOL(\cO_2:\cO_1):G_1(f)\mapsto G_2(f_{ext})$ is therefore clearly a 
symplectic map. Writing $\cA(\cO_1)$, $\cA(\cO_2:\cO_1)$ for the corresponding Weyl algebras, there must by the Bogoliubov transform exist a *-isomorphism 
$\alpha_{(2:1)}$ from $\cA(\cO_1)$ onto $\cA(\cO_2:\cO_1)$ which maps Weyl operators of the form $W(G_1(f))$ ($f \in 
\mathscr{D}(\cO_1,E)$) onto those of the form $W(G_2(f_{ext}))$. (Note that it clearly follows from \cite[Corollary 4.2.11]{BGP} that 
$\cA(\cO_2:\cO_1)$ may be regarded as a sub-algebra of $\cA(\cO_2)$.) Since by the preceding corollary we additionally have that 
$$\mu_2(G_2(f_{ext}), G_2(g_{ext})) = \mu_1(G_1(f), G_1(g))\mbox{ for all }f, g \in \mathscr{D}(\cO_1,E),$$
it is clear that the map $T$ also identifies the restriction of the two point function of $\omega_{\cO_2}$ to $\SOL(\cO_2:\cO_1)$, 
namely $$(G_2(f_{ext}), G_2(g_{ext}))\to \mu_2(G_2(f_{ext}), G_2(g_{ext}))+\frac{1}{2} \tau_2(G_2(f_{ext}), G_2(g_{ext})),$$ with 
the two point function for $\omega_{\cO_1}$, namely $$(G_1(f), G_1(g))\to \mu_1(G_1(f), G_1(g))+\frac{1}{2} \tau_1(G_1(f), G_1(g).$$ 
This then ensures that the state $\omega_{\cO_2}\circ\alpha_{(2:1)}$ on $\cA(\cO_1)$ agrees $\omega_{\cO_1}$. Here we made use of the 
fact that since $\alpha_{(2:1)}$ maps onto $\cA(\cO_2:\cO_1)$, we are in the composition $\omega_{\cO_2}\circ\alpha_{(2:1)}$ effectively 
dealing with the restriction of the state $\omega_{\cO_2}$ to $\cA(\cO_2:\cO_1)$. This leads to the following proposition: 

\begin{proposition}\label{P:defn-M(2:1)} 
\begin{enumerate}
\item There exists a *-isomorphism $\alpha_{(2:1)}$ from $\cA(\cO_1)$ onto $\cA(\cO_2:\cO_1)$ which maps Weyl operators of 
the form $W(G_1(f))$ ($f \in \mathscr{D}(\cO_1,E)$) onto those of the form $W(G_2(f_{ext}))$ and for which we additionally have that 
$(\omega_2{\upharpoonright}{\cA(\cO_2:\cO_1)})\circ\alpha_{(2:1)}=\omega_{\cO_1}$. 
\item Let $p_{(2:1)}$ be the projection from $H_{\omega_2}$ (the GNS Hilbert space of the pair $(\cM(\cO_2),\omega_2)$ where $\cM(\cO_2) =\pi_2(\cA(\cO_2))''$) onto 
the Hilbert subspace generated by $\{\eta_2(W(x))\colon x\in \SOL(\cO_2:\cO_1)\}$, and write $\cM(\cO_2:\cO_1)$ for the von Neumann 
subalgebra of $\cM(\cO_2)$ generated by the Weyl operators $\{\pi_2(W(x))\colon \SOL(\cO_2:\cO_1)\}$. Then $\cM(\cO_2:\cO_1)$ commutes with 
$p_{(2:1)}$. In addition the GNS Hilbert space of the pair $(\cA(\cO_2:\cO_1),\omega_2{\upharpoonright}{\cA(\cO_2:\cO_1)})$ may be identified with 
$p_{(2:1)}H_{\omega_2}$ and $\pi_{(2:1)}(\cA(\cO_2:\cO_1))''$ with the restriction of $\cM(\cO_2:\cO_1)$ to $p_{(2:1)}H_{\omega_2}$ where in this instance 
$\pi_{(2:1)}$ denotes the *-isomorphism forming part of the GNS construction for the pair $(\cA(\cO_2:\cO_1),\omega_2{\upharpoonright}{\cA(\cO_2:\cO_1)})$.
\end{enumerate}
\end{proposition}

\begin{proof} The first claim is a direct consequence of the preceding discussion. The identification of $p_{(2:1)}H_{\omega_2}$ with the 
GNS Hilbert space of the pair $(\cA(\cO_2:\cO_1),\omega_2{\upharpoonright}{\cA(\cO_2:\cO_1)})$ is a fairly straghtforward exercise. Once we show 
that $\cM(\cO_2:\cO_1)$ commutes with $p_{(2:1)}$, the final claim then follows by noting that strong convergence of nets in 
$\mathrm{span}\{\pi_{(2:1)}(W(x))\colon \SOL(\cO_2:\cO_1)\}$ correspond to strong convergence of the matching nets in $\mathrm{span}\{\pi_2(W(x))\colon \SOL(\cO_2:\cO_1)\}$. To see the claim 
regarding commutation with $p_{(2:1)}$ we firstly note that we will for any $x,y \in \SOL(\cO_2:\cO_1)$ have that 
$\pi_2(W(x))\eta_2(W(y)) = \eta_2(W(x)W(y))=e^{-i\tau_2(x,y)}\eta_2(W(x+y))$. Each $\pi_2(W(x))$ ($x\in \SOL(\cO_2:\cO_1)$) therefore 
maps $\overline{\mathrm{span}}\{\pi_2(W(x))\colon \SOL(\cO_2:\cO_1)\}=p_{(2:1)}H_{\omega_2}$ back into itself. So for each $x\in \SOL(\cO_2:\cO_1)$ and any 
$a, b\in H_{\omega_2}$ we have that $$\langle a, p_{(2:1)}\pi_2(W(x))p_{(2:1)}b\rangle = \langle a, \pi_2(W(x))p_{(2:1)}b\rangle$$ and 
hence that $p_{(2:1)}\pi_2(W(x))p_{(2:1)} = \pi_2(W(x))p_{(2:1)}$. Since this is true for the generators of $\cM(\cO_2:\cO_2)$ it must 
be true for all elements of $\cM(\cO_2:\cO_2)$ with the claim regarding commutation then following by duality.
\end{proof}

It remains to show that the GNS standard form of the pair  $(\cM(\cO_2:\cO_1),\omega_2{\upharpoonright}{\cM(\cO_2:\cO_1)})$ as described above, embeds 
into the GNS standard form of the pair $(\cM(\cO_2),\omega_2)$. For this task we fisrtly need the following specialization of \cite[Theorem 5.2.24(b)(i)-(iii)]{KM}.

\begin{lemma} Let $p_{(2:1)}$ and $H_{\omega_2}$ be as before. Then the following holds true:
\begin{enumerate}
\item With $\widetilde{e}_{(2:1)}$ denoting the projection of $\mathcal{H}_{\mu_2}$ onto the Hilbert subspace of $\mathcal{H}_{\mu_2}$ 
generated by $\SOL(\cO_2:\cO_1)$, the space $p_{(2:1)}H_{\omega_2}$ is the bosonic (symmetrized) Fock space $\mathcal{F}^+(\widetilde{e}_{(2:1)}\mathcal{H}_{\mu_2})= 
\Gamma(\widetilde{e}_{(2:1)})\mathcal{F}_+(\mathcal{H}_{\mu_2})$. In particular $p_{(2:1)}=\Gamma(\widetilde{e}_{(2:1)})$ and $H_{\omega_2} =\mathcal{F}_+(\mathcal{H}_{\mu_2})$.
\item With $\Psi_2$ denoting the vacuum vector of the Fock space $\mathcal{F}_+(\mathcal{H}_{\mu_2})$, the subspace of the finite complex linear combinations of $\Psi_2$ and
all vectors of the form $$a^*(x_1)\dots a^*(x_n)\Psi_2 \mbox{ for }n = 1, 2,\dots$$where $x_k\in \SOL(\cO_2:\cO_1)$ for all $1\leq k\leq n$, is a dense subspace of $p_{(2:1)}H_{\omega_2}$.
\end{enumerate}
\end{lemma}

\begin{proof} The proof follows from a straightforward modification of the first paragraph of the proof of \cite[Lemma A.2]{KW}.
\end{proof}

The above lemma now enables us to note the following:

\begin{lemma}\label{mod-subalg1} Let $p_{(2:1)}$ and $H_{\omega_2}$ be as before and consider the GNS standard form $(\cM(\cO_2),H_{\omega_2},J_{\cO_2},\mathscr{P}_{\cO_2})$. Then 
$p_{(2:1)}H_{\omega_2}$ is an invariant subspace of the modular conjugation operator $J_{\cO_2}$. 
\end{lemma}

\begin{proof} Taking into account the unitary equivalence of the Araki-Woods and Kay-Wald approached noted in subsection \ref{AWF}, the claim follows from the preceding lemma considered 
alongside the final two equalities in \cite[Theorem 40(3)]{Der}. See also the fragment of the proof of \cite[Theorem 40(6)]{Der} showing that $J_s$ preserves the space $\mathcal{H}$ defined there. 
\end{proof}

We finally arrive at the main result for this subsection.

\begin{proposition} There is a standard form homomorphism from the standard form of $(\cM(\cO_1),\omega_{\cO_1})$ into the standard form of $(\cM(\cO_2),\omega(\cO_2))$.
\end{proposition}

\begin{proof} Since by Proposition \ref{P:defn-M(2:1)} the pairs $(\cA(\cO_2:\cO_1),\omega_2{\upharpoonright}{\cA(\cO_2:\cO_1)}$ and $(\cA(\cO_1),\omega_{\cO_1})$ are *-isomorphically equivalent, 
it is clear that they will generate ismorphic standard forms for $\cM(\cO_2:\cO_1)$ and $\cM(\cO_1)$ respectively. It therefore remains to show that the standard form of the pair 
$(\cM(\cO_2:\cO_1),\omega_2{\upharpoonright}{\cM(\cO_2:\cO_1)}$ embeds into the standard form of the pair $(\cM(\cO_2),\omega_2)$. This in turn follows from a combination of part (2) of  
Proposition \ref{P:defn-M(2:1)}, and an application of Lemma \ref{mod-subalg2} to what we noted in Lemma \ref{mod-subalg1}.  
\end{proof}

\subsection{Emergent framework}

A comparison of the previous two subsections now yield the following conclusion:

\begin{theorem} With respect to the morphisms defined for the category \textbf{VN-MA-SF} introduced in section \ref{S3}, the collection of pairs 
$\{(\cM(\mathcal{O}), \omega_\mathcal{O}):\mathcal{O}\in \mathcal{K}(M,g)\}$ describes a locally covariant theory in the sense of section 
3 of \cite{FV2}. 
\end{theorem}

Since the Hadamard states  are a subset of quasi-free states, the above conclusion is obviously valid for these states. 
It is well known (see \cite{Wald2}) that  Hadamard states have well-defined expectation values of the stress-energy tensor. 
Since the description of Hadamard states would go beyond the scope of this work we will limit ourselves to only a few remarks. The key property of these states is given in the language of 
distribution theory. And this description is well adapted to the Wightman axioms (also called G{\aa}rding-Wightman axioms), i.e. field operators localized at a point of spacetime. In this paper, 
we analyze quantum fields localized in spacetime regions. Therefore, continuing the strategy described in the paper \cite{LM}, we will in the section \ref{S9} provide conditions also leading to the 
description of (regular) quantum fields with a well defined average value of the stress-energy tensor operator.

\section{States invariant under the action of a Killing vector field $Z$}\label{S8}

Let $(M,g)$ be a globally hyperbolic spacetime which is stationary, i.e. $(M,g)$ admits a time-like Killing field $Z$ as well as a one parameter group of isometries $\beta_t^Z : M \to M$ whose orbits 
are time-like. Consider the Weyl algebra $\cA(M)$ associated to a real scalar Klein-Gordon field as was described in \cite{KM}. It was shown in \cite{KM} that there exists the quasi-free state 
$\omega_Z$ on $\cA(M)$ which is invariant under the action of $Z$, i.e. $\omega_Z \circ \pi^Z_t = \omega_Z$ where $\pi^Z_t(W(\psi)) = W(\psi \circ \beta^Z_t))$; cf. Section 5.2.7 in \cite{KM}.

Consequently, taking the GNS representation $\pi_{\omega_z}(\cdot)$ of $\cA(M)$ induced by the state $\omega_Z$, see \cite{KM}, we arrive at the dynamical system
\begin{equation}
( \qM(M), \alpha^Z_t),
\end{equation}
where $\qM(M)$ is the weak closure of $\pi_{\omega_Z}(\cA(M))$, $\alpha^Z_t(\pi_{\omega_Z}(a) ) = U^{(Z)}_t \pi_{\omega_Z}(a) {U_t^{{(Z)}}}^*
= \pi_{\omega_Z}(\pi^Z_t(a))$, for all $t \in \Rn$ and $a \in \cA(M)$.

Strong continuity $t \mapsto U^{(Z)}_t$ of the unitary group $\{U^{(Z)}_t\}$ ensures the existence of an infinitesimal generator which may then be interpreted 
as the natural Hamiltonian for the considered system. We pause to justify the claim of continuity. As is pointed out in the discussion following 
equation (5.40) in \cite{KM}, strong continuity of this unitary group is equivalent to the requirement that 
\begin{equation}\label{strcon-ug} \lim_{t\to 0}\omega_Z(a^*\alpha^Z_t(a))=\omega_Z(a^*a)\mbox{ for all }a\in \cA(M). \end{equation}
This will of course follow if the automorphism group $\pi_{\omega_Z}(a)\to \pi_{\omega_Z}(\alpha^Z_t(a))$ where $a\in\cA(M)$ is strongly continuous at 0. With $\mu_Z$ denoting the inner product that 
forms the real part of the two point function of $\omega_Z$ (see \cite[Theorem 5.2.24]{KM}), we claim that this automorphism group will indeed exhibit such strong continuity if for all 
$f\in \mathscr{D}(M,E)$ we have that $\mu_Z(G(f)-G(f)\circ\beta_t^Z,G(f)-G(f)\circ\beta_t^Z)\to 0$ as $t\to 0$. Right at the outset we note that the isometries $\beta_t^Z$ commute with the Green's 
operator $G$. (See for example page 226 of \cite{Dim}.) This ensures that the action of these isometries on $\mathscr{D}(M,E)$, lift to an action on $\{G(f)\colon f\in\mathscr{D}(M,E)\}$. We shall 
constantly make silent use of this fact. Since $\cA(M)$ is generated by the Weyl operators and since the span of vectors of the form $\pi_{\omega_Z}(W(G(f)))\Omega_{\omega_Z}$ is dense in the GNS 
Hilbert space, this claim will follow if we can show that for any $f, g\in \mathscr{D}(M,E)$ we have that $\|\pi_\omega((W(G(f)\circ\beta^Z_t)-W(G(f)))W(G(g)))\Omega_\omega\|_\omega\to 0$ as $t\to 0$. 
Here we dropped the subscript $Z$ for the sake of simplicity. In the following we also set $\psi=G(f)$ and $\varphi=G(g)$ for the sake of simplicity. Now observe that 
\begin{eqnarray*}
&&\|\pi_\omega((W(\psi\circ\beta^Z_t)-W(\psi)))W(\varphi))\Omega_\omega\|^2_\omega\\
&&\quad =  \|\pi_\omega((W(\psi)^*W(\psi\circ\beta^Z_t)-\I))W(\varphi))\Omega_\omega\|^2_\omega\\ 
&&\quad =  \|\pi_\omega((W(-\psi)W(\psi\circ\beta^Z_t)-\I))W(\varphi))\Omega_\omega\|^2_\omega\\ 
&&\quad = \|\pi_\omega((e^{i\sigma(\psi,\psi\circ\beta_t^Z)/2}W(\psi\circ\beta^Z_t-\psi)-\I))W(\varphi))\Omega_\omega\|^2_\omega\\
&&\quad = \omega(2\I-e^{i\sigma(\psi,\psi\circ\beta_t^Z)/2}W(\varphi)^*W(\psi\circ\beta^Z_t-\psi)W(\varphi) -\\ 
&&\quad\qquad e^{-i\sigma(\psi,\psi\circ\beta_t^Z)/2}W(\varphi)^*W(\psi\circ\beta^Z_t-\psi)^*W(\varphi))\\
&&\quad = 2-e^{i\sigma(\psi,\psi\circ\beta_t^Z)/2}e^{i\sigma(\varphi,\psi\circ\beta_t^Z-\psi)}\omega_Z(W(\psi\circ\beta^Z_t-\psi))-\\
&&\quad\qquad e^{i\sigma(\psi,\psi\circ\beta_t^Z)/2}e^{i\sigma(\varphi,\psi-\psi\circ\beta_t^Z)}\omega_Z(W(\psi-\psi\circ\beta^Z_t))\\
&&\quad = 2-e^{i\sigma(\psi+2\varphi,\psi\circ\beta_t^Z-\psi)/2}\omega_Z(W(\psi\circ\beta^Z_t-\psi))-\\
&&\quad\qquad e^{i\sigma(2\varphi-\psi,\psi-\psi\circ\beta_t^Z)/2}\omega_Z(W(\psi-\psi\circ\beta^Z_t))\\
&&\quad = 2-e^{-i\sigma(2\varphi+\psi,\psi\circ\beta_t^Z-\psi)/2}e^{-\mu_Z(\psi-\psi\circ\beta^Z_t,\psi-\psi\circ\beta^Z_t)/2}-\\
&&\quad\qquad e^{i\sigma(2\varphi-\psi,\psi-\psi\circ\beta_t^Z)/2}e^{-\mu_Z(\psi-\psi\circ\beta^Z_t,\psi-\psi\circ\beta^Z_t)/2}
\end{eqnarray*}
(In the third last equality we used the fact that $W(\varphi)^*W(\psi)W(\varphi)= e^{-i\sigma(\varphi,\psi)}W(\psi)$ - a fact which can easily be verified from 
the properties of Weyl operators noted in Remark \ref{BGP-localg}.)  The symplectic form $\sigma$ we have to work with here is of course 
$\sigma(G(f),G(g))=\int_M\langle f, G(g)\rangle\,dV$. For this form we do have that  
$$\sigma(G(f)\circ\beta^Z_t(f),G(2g\pm f))= \int_M\langle f\circ\beta_t^Z, G(2g\pm f)\rangle\,dV \to$$ $$\int_M\langle f, G(2g\pm f)\rangle\,dV = \sigma(G(f),G(2g\pm f))\mbox{ as }t\to 0.$$So 
$\sigma(G(2g\pm f),G(f)\circ\beta_t^Z-G(f))\to 0$ as $t\to 0$. If we now combine this fact with the previously 
centred equations, it clearly follows that if for all $f\in \mathscr{D}(M,E)$ we have that $\mu_Z(G(f)-G(f)\circ\beta_t^Z,G(f)-G(f)\circ\beta_t^Z)\to 0$ as 
$t\to 0$, then $\|\pi_\omega((W(f\circ\beta^Z_t)-W(f))W(g))\Omega_\omega\|_\omega\to 0$ as $t\to 0$. As noted earlier this is sufficient to ensure that equation 
(\ref{strcon-ug}) will hold whenever $\lim_{t\to 0}\mu_Z(G(f)-G(f)\circ\beta_t^Z,G(f)-G(f)\circ\beta_t^Z)= 0$ for all $f\in \mathscr{D}(M,E)$.

\begin{remark}
It is an easy observation that the dynamical system $(\qM(M), \alpha^Z_t)$ given above provides a nice illustration of the Section \ref{KillingFlows}. Namely, 
taking derivatives of $\alpha^{(Z)}_t$ will lead to Quantum Killing Lie derivatives. Here it is the invariance $\omega_Z \circ \pi^Z_t = \omega_Z$ noted above 
that ensures that Theorem \ref{LieDer-exist} is applicable thereby ensuring the existence of these derivatives as densely defined closable *-derivations. 
\end{remark}

\section{Enlarging the observable algebras}\label{S9} 

Turning to the question  of enlarging the set of observables so that the enlarged set allows for the existence of arbitrary moments of the field operators, we 
note that $\pi_{\omega_Z}(W(f)) = e^{i\phi_{\omega_Z}(f)}$, where $\phi_{\omega_Z}(f)$ are smeared self-adjoint field operators. It is worth noting that, by the 
construction, $\{\phi_{\omega_Z}(f)\}$ are affiliated to $\qM(M)$.

We recall, that the Orlicz space technique applied to regular fields satisfying an $H$-boundedness condition with respect to the density of the dual weight of the reference 
weight $\omega$ of the von Neumann algebra, leads to the space $L^{\cosh - 1}(\qM)$. For all details see the paper \cite{LM}. The key observation to be made here, is that under 
such an $H$-boundedness restriction, the field operators affiliated to $\qM(M)$ which have all moments well defined may in a very natural way be embedded into this space. 
In other words, if the field operators $\{\phi_{\omega_Z(f)}\}$ satisfy an $H$-boundedness condition, then they can be canonically embedded into $L^{\cosh - 1}(\qM(M))$.
In that case $\langle T_{ab}\rangle_{\omega_Z}$ exists and consequently the (semiclassical) Einstein equation would be well defined. In situations where we have a strongly continuous 
action $\alpha_t$ that not only leaves $\omega$ invariant but also corresponds to the modular automorphism group of some other fns weight, one can use results like \cite[Corollary VIII.3.6]{Tak} and 
\cite[Proposition 4.2]{BL} to show that in the crossed product $\qM\rtimes_\omega\mathbb{R}$ this action is implemented by a unitary group with generator a product of the density of the dual weight and 
an operator affiliated to the centralizer of $\qM$.

It is crucial to note that the Orlicz space technique allows for the consideration of a broader class of states than quasi-free states. Namely, the natural states for $L^{cosh-1}(\qM(M))$-observables 
are normalized positive functionals corresponding to elements in the space $L\log(L+1)(\qM)$, see \cite{ML}, \cite{LM}. The above indicates that the proposed description of regular fields operators 
can be used to analyze even (strongly) interacting fields having strong quantum correlations, for example entanglements, see Subsection \ref{ss-entangle}.

We end this section with the following remarks on $H$-boundedness. In the present paper we are concerned with showing the usefulness of von Neumann algebras in 
the description of quantum systems in curved spacetime. In particular, the AQFT framework was formulated in this language. 
When discussing Hadamard states in the previous section we mentioned Wightman's axioms. It is key to note that
quantum fields described by Wightman axiomatics and satisfying the $H$-boundedness condition can be equally described in the AQFT language.
More precisely, there is a correspondence between quantum fields in the sense of Wightman with the test function space belonging to the theory of ultradistribution and regular field operators 
affiliated to the net of local algebras, see \cite{FH}, \cite{Kern}, \cite{Sum}. In this work we were concerned with the net of local von Neumann algebras and field operators affiliated to 
such nets. Therefore, the $H$-bounded condition can be considered a natural requirement.

\part{Conclusion: the emergent framework and future development}

\section{Concluding remarks}

In this work, which is a continuation and an extension of our recent paper \cite{LM}, we, firstly, present the AQFT framework in terms of von Neumann algebras. At this point we wish to 
re-emphasise that von Neumann algebras are the key component in noncommutative integration theory. The presented framework is natural from a categorical point of view.  Furthermore, the 
mentioned approach allows us to formulate the framework for AQFT without at the outset assuming the existence of a global Hilbert space. The important point to note here is that the presented 
approach allows the description of the entanglement phenomenon in QFT. In the next part of the paper, which is devoted to the study of differential calculus within the framework of field 
dynamics, we introduced the description of flows on local algebras without using the concept of ``tangentially conditioned'' algebras. In this way, we obtain a complementary and more simplified 
description of flows than that in our previous paper \cite{LM}.  Here, flows for local quantum systems in curved spacetime  are defined directly using Killing vector  fields acting on spacetime. 
The basic properties of quantum local flows have been shown. Furthermore, the conditions under which quantum Lie derivatives exist are provided. This was achieved by studying quantum local flows 
on regions of spacetime. We also showed the one may use the geometry of the underlying Lorentzian manifold itself to construct very natural quasi-free states at a local level, with the resultant 
net of local algebras satisfying the criteria of the category \textbf{VN-MA-SF}, which in turn ensures that this net is an appropriate framework for a locally covariant quantum field theory. In 
summary, the result of our investigation suggests that the following is a good framework for AQFT which avoids the imposition of an a priori global Hilbert spaces and yet leaves room for a theory of 
entanglement at a local level: 

\begin{itemize}
\item Given a globally hyperbolic Lorentzian manifold $(M,G)$, to each $\mathcal{O}\in \mathcal{K}(M,g)$ there corresponds a von Neumann measure algebra $(\cM(\cO),\omega_{\cO})$ in standard form 
(cf. Definition \ref{vnmasf-cat}) where each $\omega_{\cO}$ is a quasi-free state such that in terms of the morphisms as defined in Definition \ref{vnmasf-cat}, the identification $\cO\to(\cM(\cO),\omega_{\cO})$ yields a locally covariant quantum field 
theory for $(M,g)$ in the sense of \cite[Definition 2.1]{BFV}.
\item In the above framework the action of local flows of Killing vector fields is realised at the algebra level for short times in the sense described in Corollary \ref{loc-flow}. (In the case 
considered by Haag and Kastler, where $(M,g)$ is just Minkowski space, this property corresponds to the realisation of the action of Poincar\'e group at the algebra level.) 
\end{itemize}

The natural remaining question is whether the additivity property for quantum fields holds in the proposed description of nets of local algebras. For a deeper discussion of additivity the 
reader is refered to \cite{Robert}, \cite{SW}. We will here only note that the additivity property was very useful for the description of a quantum flows given in Theorem 4.1 in \cite{LM}. It is 
worth pointing out that for a quasi-free state satisfying mild regularity conditions, the additivity property will in the present setting also hold. To see this we need the concept of a partition 
of unity.

\begin{definition} A collection of subsets $\mathcal{U} = \{u_\alpha\colon\alpha\in A\}$ of a manifold $M$ is called
\emph{locally finite}, if for all $p \in M$ there is an neighborhood $V_p$ of $p$ with $V_p\cap U_\alpha$ non-empty for only finitely many of the sets $U_\alpha$.
\end{definition}

\begin{definition} A \emph{partition of unity} on a manifold $M$ is a collection  $\{x_\gamma\colon\gamma\in \Gamma\}$ of smooth real-valued functions such that 
\begin{enumerate}
\item $\{\mathrm{supp}(x_\gamma)\colon \gamma\in \Gamma\}$ is locally finite,
\item $x_\gamma(p)\geq 0$ for all $p\in M$ and all $\gamma\in \Gamma$,
\item $\sum_{\gamma\in \Gamma}x_\gamma(p) = 1$ for all $p\in M$
\end{enumerate}
If $\mathcal{U} = \{U_\alpha\colon\alpha\in A\}$ is an open cover of $M$ we say that a partition of unity is \emph{subordinate} to the open cover $\mathcal{U}$ if for every $\gamma$ there is an 
$\alpha$ such that $\mathrm{supp}(x_\gamma)\subseteq U_\alpha$.
\end{definition}
 
The key fact we need is the following existence theorem.

\begin{theorem}\label{par-unity} \cite[Theorem 1.11]{War} Let $M$ be a $d$-dimensional $C^k$-differentiable second countable manifold and $\mathcal{W}$ any open cover. Then $M$ admits a countable 
partition of unity subordinate to the cover $\mathcal{W}$ with the support of each function in the t of unity compact.
\end{theorem}

\begin{theorem} Let $\mathcal{O}\in\mathcal{K}(M,g)$ be given and let $\omega$ be a quasi-free state on $\cA(\mathcal{O})$ of the form described in \cite[Theorem 5.2.24]{KM} for which the inner 
product $\mu$ respects the mode of convergence introduced in \cite[Definition 4.3.6]{BGP}. Then there exists a countable family $\mathcal{V}_n$ of globally hyperbolic charts in $\mathcal{O}$ such 
that $\cup_n\mathcal{V}_n=\mathcal{O}$. Moreover $\cM(\mathcal{O})=(\cup_{n=1}^\infty\cM(\mathcal{V}_n))''$ where $\cM(\mathcal{O})$ is the double commutant of $\pi_\omega(\cA(\mathcal{O}))$ and 
each $\cM(\mathcal{V}_n)$ is the double commutant of the copy of $\cA(\mathcal{V}_n)$ inside $\pi_\omega(\cA(\mathcal{O}))$.  
\end{theorem}

\begin{proof} In the latter part of the proof we will silently make extensive use of the ideas and notation used in Section \ref{S8}. A degree of familiarity with 
the flow of ideas in Section \ref{S8} is therefore necessary to navigate the proof. For each $p\in \mathcal{O}$ there exists a chart 
$\mathcal{V}_p\subset \mathcal{O}$ of $p$ which belongs to $\mathcal{K}(M,g)$ (see \cite[Theorem 2.7]{Min}). We therefore have that 
$\cup_p\mathcal{V}_p=\mathcal{O}$. We know from Moretti's notes \cite[Definition 8.2]{Mor} that $\mathcal{O}$ is second countable and hence that there must exist 
a countable subfamily $\{\mathcal{V}_n\colon n\in \mathbb{N}\}$ which still covers $\mathcal{O}$. For these charts it is clear that 
$(\cup_{n=1}^\infty\cM(\mathcal{V}_n))''\subseteq \cM(\mathcal{O})$. We are therefore left with the challenge of proving the converse containment. It is here that 
we need the concept of a partition of unity. Aside from what we noted earlier about second countability, it is also clear from Moretti's notes (see the discussion 
following \cite[Definition 2.23]{Mor}) that globally hyperbolic Lorentzian manifolds satisfy all the prerequisites of Theorem \ref{par-unity}. Hence there exists 
a countable partition of unity $\{x_n\colon n\in \mathbb{N}\}$ on $\mathcal{O}$ with compact supports subordinate to the family 
$\{\mathcal{V}_n\colon n\in \mathbb{N}\}$. 

For the next part of the proof we need to understand in which manner the $\mathrm{CCR}$ $C^*$-algebras $\cA(\mathcal{V}_n)$ live ``inside'' $\cA(\mathcal{O})$. Given 
$f, g\in \mathscr{D}(\mathcal{V}_n, E)$ and denoting the extensions to $\mathcal{O}$ by $f_{ext}$, $g_{ext}$, the fact that the supports are in $\mathcal{V}_n$ 
ensures that $$\int_M \langle f, \widetilde{G}_{\mathcal{V}_n}(g)\rangle\,dV = \int_M \langle f_{ext}, \widetilde{G}_{\mathcal{O}}(g_{ext})\rangle\,dV.$$(Here we 
used the fact that $\widetilde{G}_{\mathcal{V}_n}(g) = \widetilde{G}_{\mathcal{O}}(g_{ext}){\upharpoonright}{\mathcal{V}_n}$.) The left and right hand side of the displayed 
equation above each respectively represent degenerate symplectic forms on the spaces $\mathcal{E}(\mathcal{V}_n)=\{f\colon f\in \mathscr{D}(\mathcal{V}_n,E)\}$ 
and $\mathcal{E}(\mathcal{O}:\mathcal{V}_n)\{a\colon a\in \mathscr{D}(\mathcal{O},E), \mathrm{supp}(a)\subset \mathcal{V}_n\}$. By \cite[Lemma 4.3.8]{BGP} the map 
$\widetilde{G}_{\mathcal{V}_n}(f)\to \widetilde{G}_{\mathcal{O}}(f_{ext})$ will map $\mathrm{ker}(\widetilde{G}_{\mathcal{V}_n})$ into 
$\mathrm{ker}(\widetilde{G}_{\mathcal{O}})$ and hence induces a well-defined map from $\mathcal{E}(\mathcal{V}_n)/\mathrm{ker}(\widetilde{G}_{\mathcal{V}_n})$ to 
the subspace $\mathcal{E}(\mathcal{O}:\mathcal{V}_n)/\mathrm{ker}(\widetilde{G}_{\mathcal{O}})$ of 
$\mathcal{E}(\mathcal{O})/\mathrm{ker}(\widetilde{G}_{\mathcal{O}})$. We remind the reader that the spaces  
$\mathcal{E}(\mathcal{V}_n)/\mathrm{ker}(\widetilde{G}_{\mathcal{V}_n})$ and $\mathcal{E}(\mathcal{O})/\mathrm{ker}(\widetilde{G}_{\mathcal{O}})$ are respectively 
canonically bijective to $\mathrm{SOL}(\mathcal{V}_n)=\{\widetilde{G}_{\mathcal{V}_n}(f)\colon f\in \mathscr{D}(\mathcal{V}_n,E)\}$ and 
$\mathrm{SOL}(\mathcal{O}:\mathcal{V}_n)=\{\widetilde{G}_{\mathcal{O}}(a)\colon a\in \mathscr{D}(\mathcal{O},E), \mathrm{supp}(f)\subset \mathcal{V}_n\}$. In view of the 
equality in the previously displayed equation it is now clear that induced map from $\mathcal{E}(\mathcal{V}_n)/\mathrm{ker}(\widetilde{G}_{\mathcal{V}_n})$ to 
the subspace $\mathcal{E}(\mathcal{O}:\mathcal{V}_n)/\mathrm{ker}(\widetilde{G}_{\mathcal{O}})$ of 
$\mathcal{E}(\mathcal{O})/\mathrm{ker}(\widetilde{G}_{\mathcal{O}})$ is a symplectic isomorphism between these two spaces. By means of the aforementioned 
bijective correspondences, it may of course equivalently be regarded as a symplectic isomorphism from 
$\mathrm{SOL}(\mathcal{V}_n)=\{\widetilde{G}_{\mathcal{V}_n}(f) \colon f\in \mathscr{D}(\mathcal{V}_n,E)\}$ to 
$\mathrm{SOL}(\mathcal{O}:\mathcal{V}_n)=\{\widetilde{G}_{\mathcal{O}}(f) \colon f\in \mathscr{D}(\mathcal{O},E)\}$. Hence the $\mathrm{CCR}$ $C^*$-algebras generated by these two 
sets are $*$-isomorphic. So when regarded as a subalgebra of $\cA(\mathcal{O})$, we may identify $\cA(\mathcal{V}_n)$ with the $\mathrm{CCR}$ algebra generated by  
$\mathrm{SOL}(\mathcal{O}:\mathcal{V}_n)$. In the appropriate representation we then take the double commutants of these algebras to get the corresponding 
statement for von Neumann algebras.   

Given any $f\in\mathscr{D}(\mathcal{O}, E)$, we now set $f_n=f.x_n$. Checking shows that $f=\sum_{n=1}^\infty f_n$. Given any $f_n$, we know from the preceding 
theorem that there exists some $k\in \mathbb{N}$ with $\mathrm{supp}(f_n)\subset \mathcal{V}_k$. Checking now shows that $f_n= (g_n)_{ext}$ where 
$g_n = (f_n){\upharpoonright}{\mathcal{V}_k}$. It is clear that $g_n$ belongs to $\mathscr{D}(\mathcal{V}_k,E)$, and hence that the Weyl operator 
$W(\widetilde{G}_\mathcal{O}(f_n))$ belongs to $\cM(\mathcal{V}_k)$. It is also clear from the properties of the Weyl operators that the Weyl operator 
$W(\sum_{n=1}^kf_n)$ is a just a multiple of the product $\Pi_{n=1}^kW(f_n)$. To conclude we therefore need to show that in the representation engendered by 
$\omega$, $W(f)$ is the $\sigma$-strong limit of the sequence $(W(\sum_{n=1}^kf_n))_{k=1}^\infty$. 

For any compact subset $K$ of $M$, the local finiteness of the partition of unity ensures that for any $p\in K\cap\mathrm{supp}(f)$ we can find a neighbourhood 
$U_p\subset \mathcal{O}$ of $p$ with $U_p\cap \mathrm{supp}(f_n)$ non-empty for at most finitely many of $f_n$'s. By compactness we may cover 
$K\cap\mathrm{supp}(f)$ with finitely many such neighbourhoods, from which it then follows that 
there must exist some $N_K\in \mathbb{N}$ such that $(\sum_{n=1}^kf_n){\upharpoonright}{K}=f{\upharpoonright}{K}$ for all $k\geq N_K$. A similar argument shows that the same is true for any of 
the derivatives of the $f$ and the partial sums $\sum_{n=1}^kf_n$. This of course means $(\sum_{n=1}^kf_n)_{k=1}^\infty$ and each matching sequence of derivatives 
converges uniformly to $f$ each of its corresponding derivatives on any compact subset of $\mathcal{O}$. This has two consequences:

Firstly for any $g\in \mathrm{SOL}(\mathcal{O}:\mathcal{V}_n)$ we have that $\int_M\langle \sum_{n=1}^kf_n, \widetilde{G}_{\mathcal{O}}(g)\rangle\,dV \to \int_M\langle f, 
\widetilde{G}_{\mathcal{O}}(g)\rangle\,dV$ as $k\to \infty$. To see this note that integration with respect to the volume form can be shown to make sense for 
$C_0(M)$ (the continuous functions vanishing at infinity). On appealing to the Riesz representation theorem it now follows that integration with respect to the 
volume form may be regarded as integration with respect to a Radon measure. Since $\mathcal{O}$ has finite measure with respect to this Radon measure (due to the 
compactness of its closure), the claim now follows from basic facts regarding Radon measures and what we noted above.
Secondly since $f$ and each $\sum_{n=1}^kf_n$ is supported on the compact set $\mathrm{supp}(f)$, it follows that here $(\sum_{n=1}^kf_n)_{k=1}^\infty$ converges 
to $f$ in the mode described in \cite[Definition 3.4.6]{BGP}. But by \cite[Proposition 3.4.8]{BGP}, $(\widetilde{G}_\mathcal{O}(\sum_{n=1}^kf_n))_{k=1}^\infty$ 
will then converge to $\widetilde{G}_\mathcal{O}(f)$ in exactly the same mode. If therefore on $\mathrm{SOL}(\mathcal{O})$ the inner product $\mu$ respects this mode of 
convergence, we will have that $\mu(\widetilde{G}_\mathcal{O}(f)-\widetilde{G}_\mathcal{O}(\sum_{n=1}^kf_n), \widetilde{G}_\mathcal{O}(f)-\widetilde{G}_\mathcal{O}(\sum_{n=1}^kf_n))\to 0$ 
as $k\to \infty$. The on replacing $G(f)\circ\beta_t^Z$ with $\widetilde{G}_\mathcal{O}(\sum_{n=1}^kf_n)$, the same type of argument as was used in Section \ref{S8} now shows that in the 
representation engendered by $\omega$, $W(f)$ is, as required, the $\sigma$-strong limit of the sequence $(W(\sum_{n=1}^kf_n))_{k=1}^\infty$. 
\end{proof}

All that makes it legitimate to form
\begin{conjecture} Let $\mathcal{O}\in\mathcal{K}(M,g)$. Write $\mathrm{QuantLie}$ for the set $$\{\delta_Z\colon Z\mbox{ a Killing vector field}\}.$$Then 
$\cM^{\infty}(\mathcal{O},\mathrm{QuantLie})$ is weak* dense in $\cM(\mathcal{O})$.  
\end{conjecture}

If the conjecture is true, one can use the ideas of du Bois-Violette \cite{dje} to develop non-commutative differential geometric structures for QFT, c.f. 
the discussion in \cite{LM}.

\section{Declarations}
\subsection{Funding} The authors, Proff LE Labuschagne and WA Majewski, declare that they received no funding for this project. 
\subsection{Author contributions} Although there was some overlap, the bulk of the mathematical content was contributed by Prof LE Labuschagne and the bulk of the physical content by Prof WA Majewski.
\subsection{Data Availability Statement} This manuscript has no associated data.
\subsection{Conflict of Interest} The authors, Proff LE Labuschagne and WA Majewski, declare that they have no conflicts of interest regarding the subject matter or materials discussed 
in this manuscript.
\subsection{Clinical Trial Number} Not applicable.

\end{document}